\documentclass[runningheads]{llncs}

\usepackage[T1]{fontenc} %ugly, but required
\usepackage{citesort}
\usepackage{multicol}
\usepackage{amssymb}
\usepackage{amsmath,bbm}
\usepackage{latexsym,epic,eepic}
\usepackage{complexity}
\usepackage{multirow}
\usepackage{url,hyperref}
\usepackage{btran}
\usepackage{paralist}
\usepackage{tikz}
\usetikzlibrary{shapes}
\usetikzlibrary{automata,positioning}

\usepackage{graphicx}                   % necessary for inserting .pdfs 
\usepackage{hyperref}

\usepackage{listings}

\lstdefinelanguage[JML]{Java}[]{Java}%
       {% C++ style comments have to start with a blank!
        comment=[l]{//\ },
        morecomment=[s]{/*\ }{*/},        
        morecomment=[s]{/**}{*/},
        classoffset=1,
        morekeywords={abrupt_behavior,abrupt_behaviour,
         accessible,accessible_redundantly,also,assert,assert_redundantly,
         assignable,assignable_redundantly,assume,assume_redundantly,
         axiom,behavior,behaviour,breaks,breaks_redundantly,
         callable,callable_redundantly,captures,captures_redundantly,
         choose,choose_if,code,code_bigint_math,code_java_math,
         code_safe_math,constraint,constraint_redundantly,constructor,
         continues,continues_redundantly,decreases,decreases_redundantly,
         decreasing,decreasing_redundantly,diverges,diverges_redundantly,
         duration,duration_redundantly,ensures,ensures_redundantly,
         example,exceptional_behavior,exceptional_behaviour,
         exceptional_example,exsures,exsures_redundantly,extract,field,
         forall,for_example,ghost,helper,hence_by,hence_by_redundantly,
         implies_that,in,in_redundantly,initializer,initially,instance,
         invariant,invariant_redundantly,loop_invariant,
         loop_invariant_redundantly,maintaining,maintaining_redundantly,
         maps,maps_redundantly,measured_by,measured_by_redundantly,method,
         model,model_program,modifiable,modifiable_redundantly,modifies,
         modifies_redundantly,monitored,monitors_for,non_null,
         normal_behavior,normal_behaviour,normal_example,nowarn,
         nullable,nullable_by_default,old,or,post,post_redundantly,
         pre,pre_redundantly,pure,readable,refine,refines,refining,represents,
         represents_redundantly,requires,requires_redundantly,
         returns,returns_redundantly,set,signals,signals_only,
         signals_only_redundantly,signals_redundantly,spec_bigint_math,
         spec_java_math,spec_protected,spec_public,spec_safe_math,
         static_initializer,uninitialized,unreachable,weakly,
         when,when_redundantly,working_space,working_space_redundantly,
         writable
        },
        morekeywords={rep,peer,readonly},
        keywordsprefix=\\,
        otherkeywords={<:,<-,->,..,<==,==>,<==>,<=!=>},
        classoffset=0 % restore default class for keywords
}

\newenvironment{to-do}
{ \rule{1ex}{1ex}\hspace{\stretch{1}} \bfseries}
{ \hspace{\stretch{1}}\rule{1ex}{1ex} \vspace{1ex}}

\newcommand{\defn}[1]{\textit{#1}}

\newcommand{\Aut}{\ensuremath{\mathcal{A}}} %Automata
\newcommand{\Tra}{\ensuremath{\mathcal{T}}} %Transducer
\newcommand{\transrel}{\ensuremath{\Delta}} %Automata Transition Relation

\newcommand{\controls}{\ensuremath{Q}} %Automata Transition Relation
\newcommand{\finals}{\ensuremath{F}} % Set of final states
\newcommand{\params}{\mathcal{X}}

\newcommand{\subseteqf}{{\subseteq_{\text{fin}}}}
\newcommand{\curr}{{\textit{curr}}}

\newcommand{\OMIT}[1]{}

\DeclareMathOperator{\type}{type}

\newcommand{\seqSolver}{{\scshape SeCo}}

\newcommand{\theory}{T}
\newcommand{\Constants}{\mathcal{C}}
\newcommand{\Variables}{\mathcal{V}}

\usepackage{diego-macro}
\renewcommand{\phi}{\varphi}%

\newcommand{\struct}{\mathfrak{S}}
\newcommand{\Xseq}{\textbf{x}}
\newcommand{\Yseq}{\textbf{y}}
\newcommand{\Zseq}{\textbf{z}}
\newcommand{\Wseq}{\textbf{w}}
\newcommand{\vocab}{\sigma}

\newif\ifextendedtr
\extendedtrfalse

\usepackage[backgroundcolor=orange!50, textsize=scriptsize,textwidth=1cm]{todonotes}

\newcommand{\sideanthony}[1]{\todo[backgroundcolor=green!20]{{\bf A} #1}}

\definecolor{light-gray}{gray}{0.9}
\definecolor{light-yellow}{RGB}{255, 255, 220}

\usepackage{cite}

\newcommand\shortlong[2]{#1}

\usepackage{color}
\shortlong{}{\pagestyle{plain}}

\title{Decision Procedures for Sequence Theories (Technical Report)\thanks{
A.\ Jeż was supported under National Science Centre, Poland project number
2017/26/E/ST6/00191. A.\ Lin and O.\ Markgraf were supported by the ERC
Consolidator Grant 101089343 (LASD).
P.\ R\"ummer was supported by
    the Swedish Research Council (VR)
    under grant~2018-04727, the Swedish Foundation for Strategic
    Research (SSF) under the project WebSec (Ref.\ RIT17-0011), and the
    Wallenberg project UPDATE.}}

\begin{document}

\author{Artur Je\.z\inst{1}\orcidID{0000-0003-4321-3105} \and
Anthony W.~Lin\inst{2,3}\orcidID{0000-0003-4715-5096} \and
Oliver Markgraf\inst{2}\orcidID{0000-0003-4817-4563} \and
Philipp R\"ummer\inst{4,5}\orcidID{0000-0002-2733-7098}
}

\institute{
University of Wroc\l{}aw, Poland\and
TU Kaiserslautern, Germany\and
Max Planck Institute for Software Systems,\newline Kaiserslautern, Germany\and
University of Regensburg, Germany\and
Uppsala University, Sweden
}

\renewcommand{\thelstlisting}{\arabic{lstlisting}}

\maketitle

\begin{abstract}
    Sequence theories are an extension of theories of strings with an infinite 
alphabet of letters, together with a corresponding alphabet theory (e.g.\ linear
integer arithmetic).
Sequences are natural abstractions of %the standard data type 
%including Python 
%list, Java ArrayList, and JavaScript arrays, which are 
extendable arrays, which
permit a wealth of operations including append, map, split, and concatenation.
%\textit{etc}. 
In spite of the growing amount of tool support for theories of 
sequences by leading SMT-solvers, little is known about the decidability of
sequence theories, which is in stark contrast to the state of the theories of
strings. 
%This paper performs the first systematic investigation of the
%decidability and computational complexity of sequence theories. 
We show that the decidable theory
of strings with concatenation and regular constraints can be extended to the
world of sequences over an alphabet theory that forms a Boolean algebra, while 
preserving decidability. In particular, decidability holds when regular 
constraints are interpreted as parametric automata (which extend both symbolic
automata and variable automata), but fails when interpreted as register 
automata (even over the alphabet theory of equality).
%and complexity in the decidable cases ranges
%between \pspace and \expspace, depending on the theory. 
When length constraints are added, the problem is Turing-equivalent to word 
equations with length (and regular) constraints. Similar investigations are 
conducted in the presence of symbolic transducers, which naturally model 
sequence functions
like map, split, filter, \textit{etc}. We have developed a new sequence solver,
{\seqSolver}, based
on parametric automata, and show its efficacy on two classes of benchmarks:
(i)~invariant checking on array-manipulating programs and parameterized systems,
and (ii)~benchmarks on symbolic register automata.

\OMIT{
To 
theory of sequences with 
concatenation with extended with constraints from a decidable alphabet theory that 
forms a boolean algebra is decidable (which includes most existing SMT 
theories). Just as in the case of strings, we study the extension with ``regular
pattern matching''. 
}
%
%over which an alphabet theory is defined. 

\end{abstract}

%!TEX root = cav23.tex
\section{Introduction}
\label{sec:intro}

Sequences are an extension of strings, wherein elements might range 
over an infinite domain (e.g., integers, strings, and even sequences themselves).
Sequences are ubiquitous and commonly used data types in modern 
programming languages. They come under different names, e.g.,
Python/Haskell/Prolog lists, Java ArrayList (and to some extent Streams) and 
JavaScript arrays. Crucially, sequences are \emph{extendable}, and a
plethora
of operations (including append, map, split, filter, concatenation, etc.) can
naturally be defined and are supported by built-in library functions in most
modern programming languages. 

Various techniques in software model checking \cite{rupak-survey} --- including
symbolic execution, invariant generation --- require an appropriate SMT theory,
to which verification conditions could be discharged. In the case of programs
operating on sequences, we would consequently require an SMT theory of
sequences, for which leading SMT solvers like Z3~\cite{Z3,Z3-programming} and cvc5~\cite{CVC5} 
already provide some basic support for over a decade. The basic design 
of sequence theories, as done in Z3 and cvc5, as well
as in other formalisms like symbolic automata~\cite{symbolic-power},
is in fact quite natural. 
That is,
sequence theories can be thought of as extensions of theories of strings with 
an infinite alphabet of letters, together with a corresponding alphabet theory,
e.g.\ Linear Integer Arithmetic (LIA) for reasoning about sequences of integers.
Despite this, very little is known about what is decidable over theories of
sequences.

\OMIT{
This still leaves quite a lot of details to be filled, e.g., what the meaning of
a \emph{regular constraint} is in the setting of sequences.
In the case of string
theories, regular constraints are simply of the form $x \in L(E)$,
where $E$ is a regular expression, mandating that the expression~$E$
matches the string represented by $x$. 
%Sequence theory solvers have,
%among others, been applied in the verification of smart contracts (e.g.~\cite{sequence-solidity}).
}

In the case of finite alphabets, sequence theories become theories over 
strings, in which a lot of progress has been
made in the last few decades, barring the long-standing open problem of string
equations with length constraints (e.g.\ see \cite{ganesh-word}). For example, it
is known that the existential theory of concatenation over strings with
regular constraints is decidable (in fact, \pspace-complete), e.g., see
\cite{wordequations,plandowskistoc2,diekert,schulz90,Makanin}. Here, a \emph{regular
constraint} takes the form $x \in L(E)$, 
where $E$ is a regular expression, mandating that the expression~$E$
matches the string represented by $x$. 
In
addition, several natural syntactic restrictions --- including 
straight-line, acylicity, and chain-free
(e.g.~\cite{ganesh-word,graph-rational,cav14-string,LB16,popl19,chain-free,priorities}) ---
have been identified, with which string constraints remain decidable in the
presence of more complex string functions (e.g.\ transducers, replace-all,
reverse, etc.).
In the case of infinite alphabets, only a handful of results are available.
Furia~\cite{Furia10} showed that the existential theory of sequence equations 
over the alphabet theory of %Linear Integer Arithmetic (LIA)
LIA is decidable by a 
reduction to the existential theory of concatenation over strings (over a finite
alphabet) \emph{without regular constraints}. Loosely speaking, a number (e.g.~4)
can be represented as a string in unary (e.g.~1111), and addition is then 
simulated by concatenation. Therefore, his decidability result does not extend to
other data domains and alphabet theories.
Wang et al.~\cite{DBLP:journals/jar/WangA23} define an extension of
the array property fragment~\cite{DBLP:conf/vmcai/BradleyMS06} with
concatenation.  This fragment imposes strong restrictions, however, on
the equations between sequences (here called finite arrays) that can be
considered.

\OMIT{
On the foundational side, the current state of affairs regarding sequence
theories is far from satisfactory:
virtually no decidable fragment is known and very little is known about
undecidable. In the case of finite alphabets, the problem boils
down to theories over strings, on which a lot of progress has been
made in the last few decades, barring the long-standing open problem of string
equations with length constraints (e.g.\ see \cite{ganesh-word}). For example, it
is known that the existential theory of concatenation over strings with
regular constraints is decidable (in fact, \pspace-complete), e.g., see
\cite{wordequations,plandowskistoc2,diekert,schulz90,Makanin}. In
addition, several natural syntactic restrictions --- including 
straight-line, acylicity, and chain-free
(e.g.~\cite{ganesh-word,graph-rational,cav14-string,LB16,popl19,chain-free}) ---
have been identified, with which string constraints remain decidable in the
presence of more complex string functions (e.g.\ transducers, replace-all,
reverse, etc.).
In the case of infinite alphabets, only a handful of results are available.
Furia~\cite{Furia10} showed that the existential theory of sequence equations 
over the alphabet theory of %Linear Integer Arithmetic (LIA)
LIA is decidable by a 
reduction to the existential theory of concatenation over strings (over a finite
alphabet) \emph{without regular constraints}. Loosely speaking, a number (e.g.~4)
can be represented as a string in unary (e.g.~1111), and addition is then 
simulated by concatenation. Therefore, his decidability result does not extend to
other data domains and alphabet theories.
Wang et al.~\cite{DBLP:journals/jar/WangA23} define an extension of
the array property fragment~\cite{DBLP:conf/vmcai/BradleyMS06} with
concatenation.  This fragment imposes strong restrictions, however, on
the equations between sequences (here called finite arrays) that can be
considered.
In addition, there is so far
no consistent definition of the notion of a
``regular language'' over infinite alphabets
(e.g.\ see \cite{atom-book}).
}

%Modern programming languages (e.g.\ Python, Java, JavaScript) support extendable
%arrays.

\paragraph{``Regular constraints'' over sequences.} 
%There is also the question of what is a regular constraint over sequences
%and which definition permits decidability of sequence theories. 
One answer of
what a regular constraint is over sequences is provided by \emph{automata 
modulo theories}.
\OMIT{
In the case of infinite alphabet, there is so far also no consistent 
definition of the notion of a ``regular language'' over infinite alphabets
(e.g.\ see \cite{atom-book}). This yields another question of what definitions
of regular constraints permit decidable sequence theories. 
}
Automata modulo 
theories \cite{DV21,symbolic-power} are an elegant framework that can be
used to capture the 
notion of regular constraints over sequences:
Fix an alphabet theory $\theory$ 
that forms a Boolean algebra; this is satisfied by virtually all existing SMT 
theories. 
%We start by discussing the constraint
%language design and, in particular, what should be used in place of ``regular
%constraints''. 
In this framework, one uses formulas in $\theory$
to capture multiple (possibly infinitely many) transitions of an automaton.
%One natural definition of regular
%constraints would be the notion of symbolic automata of Veanes et al.
%\cite{DV21,symbolic-power}, which satisfies practically all properties of
%regular languages (over strings with a finite alphabet). 
More precisely, %one admits
%infinitely many possible transitions 
between two states in a \emph{symbolic automaton} one associates a unary\footnote{This can be generalized to any
arity, which has to be set uniformly for the automaton.} formula 
$\varphi(x) \in \theory$. For example, $q \to_{\varphi} q'$ with $\varphi := x
\equiv 0 \pmod{2}$ over %Linear Integer Arithmetic
LIA corresponds to all
transitions $q \to_i q'$ with any even number~$i$. Despite their nice
properties, it is known that many simple languages
cannot be captured using symbolic automata; e.g., one cannot express the
language consisting of sequences containing the same even number~$i$
\emph{throughout} the sequence. 

There are essentially two (expressively incomparable) extensions of symbolic 
automata that address the aforementioned problem:
\begin{inparaenum}[(i)]
\item 
Symbolic
Register Automata (SRA) \cite{SRA} and
\item Parametric Automata (PA)
\cite{FK20,FL22,FJL22}.
\end{inparaenum}
The model SRA
was obtained by combining register automata \cite{KF94} and symbolic automata.
The model PA extends symbolic automata by allowing \emph{free variables} (a.k.a.
\emph{parameters}) in the transition guards, i.e., the guard will be of the form
$\varphi(x,\bar p)$, for parameters $\bar p$. In an accepting path of PA, a
parameter~$p$ used in multiple transitions has to be instantiated
with the same value, which enables comparisons of different positions
in an input sequence. For example, we can assert that only sequences of the form
$i^*$, for an even number $i$, are accepted by the PA with a single transition
$q \to_\varphi q$ with $\varphi(x,p) := x = p \wedge x \equiv 0 \pmod{2}$ and
$q$ being the start and final state. PA can also be construed as an extension of
both variable automata \cite{GKS10} and symbolic automata. SRA and PA are not
comparable: while parameters can be construed as read-only registers, 
SRA can only compare two different positions using equality, while PA may
use a general formula in the theory in such a comparison (e.g., order).

\paragraph{Contributions.} The main contribution of this paper is to provide
\emph{the
first decidable fragments of a theory of sequences parameterized in the element
theory}. In particular, we show how to leverage string solvers to solve theories
over sequences. We believe this is especially interesting, in view of the
plethora of existing string solvers developed in the last 10 years
(e.g. see the survey \cite{string-survey}). This opens up new possibilities for
verification tasks to be automated; in particular, we show how verification
conditions for Quicksort, as well as Bakery and Dijkstra protocols, can be
captured in our sequence theory. This formalization was done in the style of
\emph{regular model checking} \cite{RMC,RMC-revisited}, whose extension to 
infinite alphabets has been a longstanding challenge in the field.
%Based on some of these decision 
%procedures, 
We also provide a new (dedicated) sequence solver {\seqSolver}
%sequence theo
%and demonstrated its efficacy on these benchmarks, among others.
We detail our results below.
%interesting sequence
%constraints benchmarks. 
%to obtain decidable theories of sequences. As a result, 

\OMIT{
a first
systematic investigation of the decidability and complexity of sequence
theories, in particular, first decidability results.
}

We first show that the quantifier-free theory of sequences with concatenation
and PA as regular constraints is decidable.
%Artur: No reason to discuss this here.
%Since PA is not closed under
%complementation, we have to assume that each occurrence of a regular constraint
%is \emph{not negated} in the formula.
Assuming that the theory is solvable
in \pspace (which is reasonable for most SMT theories), we show that our
algorithm runs in \expspace (i.e., double-exponential time and exponential space).
We also identify conditions on the SMT theory $\theory$
under which
%exponentially better complexity, i.e.\
\pspace can be achieved
and as an example show that Linear Real Arithmetic (LRA) satisfies those conditions.
This matches the \pspace-completeness
of the theory of strings with concatenation and regular
constraints~\cite{diekertfreegroups}.

We consider three different variants/extensions:

\begin{enumerate}[(i)]
\item \emph{Add length constraints}.
Length constraints (e.g., $|\Xseq| = |\Yseq|$ for two
sequence variables $\Xseq,\Yseq$) are often considered in the context of string
theories, but the decidability of the resulting theory (i.e., strings with
concatenation and length constraints) is still a long-standing open problem
\cite{ganesh-word}. We show that the case for sequences is Turing-equivalent to
the string case. 
\item \emph{Use SRA instead of PA}. We show that the resulting theory of
sequences is undecidable, even over the alphabet theory $\theory$ of equality.
\item \emph{Add symbolic transducers}. Symbolic transducers 
\cite{DV21,symbolic-power} extend finite-state input/output
transducers in the same way that symbolic automata extend finite-state automata.
To obtain decidability, we consider formulas satisfying the straight-line
restriction that was 
defined over strings theories~\cite{LB16}. We show that the resulting theory is
decidable in 2-{\exptime} and is \expspace-hard, if $\theory$ is solvable in \pspace.
\end{enumerate}

We have implemented the solver {\seqSolver} based on our algorithms, and
demonstrated its efficacy on two classes of benchmarks:
(i) invariant checking on array-manipulating programs and parameterized systems,
and (ii) benchmarks on Symbolic Register Automata (SRA) from \cite{SRA}.
For the first benchmarks, we model as sequence constraints invariants for
QuickSort, Dijkstra's Self-Stabilizing Protocol \cite{dijkstra74} and Lamport's
Bakery Algorithm \cite{lamport74}.
%The second class of benchmarks are
%compare our solver to algorithms for 
%Symbolic Register Automata (SRA) benchmarks from \cite{SRA}, 
For (ii), we solve decision problems for SRA on benchmarks of \cite{SRA} 
such as emptiness, equivalence and inclusion on regular expressions with 
back-references.
We report promising experimental results: our solver 
{\seqSolver} is up to three orders of magnitude faster than the SRA solver in
\cite{SRA}.

\paragraph{Organization.}
We provide a motivating example of sequence theories in Section \ref{sec:mot_ex}.
%In Sec \ref{sec:prelim}, we recall some terminologies, notation, and results from 
%string constraints.
Section \ref{sec:model} contains the syntax and semantics 
of the sequence constraint language, as well as some basic algorithmic results. 
We deal with equational and regular constraints in Section
\ref{sec:we}. In Section \ref{sec:sl}, we deal with the decidable fragments with
equational constraints, regular constraints, and transducers. We deal
with extensions of these languages with length and SRA constraints in Section
\ref{sec:ext}. In Section \ref{sec:impl} we report our implementation and
experimental results. We conclude in Section \ref{sec:concl}. Missing
details and proofs can be found in the full version.

%!TEX root = cav23.tex
\section{Motivating Example}
\label{sec:mot_ex}

\begin{lstlisting}[language={[JML]Java},frame=single,float=tbp,label=lst:quicksort,captionpos=b,caption={Implementation of QuickSort with Java Streams.}]
/*@
 * ensures \forall int i; \result.contains(i) == l.contains(i);
 */
public static List<Integer> quickSort(List<Integer> l) {
  if (l.size() < 1) return l;
  Integer p = l.get(0);
  List<Integer> left   = l.stream().filter(i -> i < p)
                          .collect(Collectors.toList());
  List<Integer> right  = l.stream().skip(1).filter(i -> i >= p)
                          .collect(Collectors.toList());
  List<Integer> result = quickSort(left);
  result.add(p); result.addAll(quickSort(right));
  return result;
}
\end{lstlisting}

We illustrate the use of sequence theories in verification using a
implementation of QuickSort~\cite{DBLP:journals/cj/Hoare62}, shown in
Listing~\ref{lst:quicksort}. The example uses the Java Streams API and
resembles typical implementations of QuickSort in functional
languages; the program uses high-level operations on
streams and lists like \emph{filter} and \emph{concatenation}.
As we show, the data types and operations can naturally be modelled
using a theory of sequences over integer arithmetic, and our
results imply decidability of checks that would be done by a verification
system.

The function \lstinline!quickSort! processes a given
list~\lstinline!l! by picking the first element as the
pivot~\lstinline!p!, then creating two sub-lists \lstinline!left!,
\lstinline!right! in which all numbers \lstinline!$\geq$p! (resp.,
\lstinline!$<$p!) have been eliminated. The function
\lstinline!quickSort!  is then recursively invoked on the two
sub-lists, and the results are finally concatenated and returned.

We focus on
the verification of the post-condition shown in the beginning of
Listing~\ref{lst:quicksort}: sorting does not change the set of
elements contained in the input list. This is a weaker form of the
permutation property of sorting algorithms, and as such known to be
challenging for verification methods (e.g.,
\cite{DBLP:conf/ifm/SafariH20}). Sortedness of the result list can be
stated and verified in a similar way, but is not considered here.
Following the classical
design-by-contract approach~\cite{DBLP:journals/computer/Meyer92}, to
verify the partial correctness of the function it is enough to show
that the post-condition is established in any top-level call of the
function, assuming that the post-condition holds for all recursive
calls. For the case of non-empty lists, the verification condition,
expressed in our logic, is:
\begin{multline*}
  \left(
    \begin{array}{@{}l@{}}
      \mathbf{left} = T_{<\mathbf{l}_0}(\mathbf{l}) \wedge \mathbf{right} = T_{\geq \mathbf{l}_0}(\mathit{skip}_1(\mathbf{l})) \wedge\mbox{}
      \\
      \forall i.\, (i \in \mathbf{left} \leftrightarrow i \in \mathbf{left}')
      \wedge \forall i.\, (i \in \mathbf{right} \leftrightarrow i \in \mathbf{right}') \wedge \mbox{}
      \\
      \mathbf{res} = \mathbf{left}' \,.\, [\mathbf{l}_0] \,.\, \mathbf{right}'
    \end{array}
  \right)
  \\
  \to \forall i.\, (i \in \mathbf{l} \leftrightarrow i \in \mathbf{res})
\end{multline*}

The
variables~$\mathbf{l}, \mathbf{res}, \mathbf{left}, \mathbf{right},
\mathbf{left}', \mathbf{right}'$ range over sequences of integers,
while $i$ is a bound integer variable.
The formula uses several operators that a useful sequence theory has
to provide:
\begin{inparaenum}[(i)]
\item $\mathbf{l}_0$: the first element of input list~$\mathbf{l}$;
\item $\in$ and $\not\in$: membership and non-membership of an integer in a
  list, which
  can be expressed using symbolic parametric automata;
\item $\mathit{skip}_1$, $T_{<\mathbf{l}_0}$, $T_{\geq \mathbf{l}_0}$:
  sequence-to-sequence functions, which can be represented using
  symbolic parametric transducers;
\item $\cdot \,.\, \cdot$: concatenation of several sequences.
\end{inparaenum}
The formula otherwise is a direct model of the method in
Listing~\ref{lst:quicksort}; the
variables~$\mathbf{left}', \mathbf{right}'$ are the results of the
recursive calls, and concatenated to obtain the result sequence.

In addition, the formula contains quantifiers. To demonstrate validity
of the formula, it is enough to eliminate the last
quantifier~$\forall i$ by instantiating with a Skolem symbol~$k$, and
then instantiate the other quantifiers (left of the implication) with
the same $k$:
\begin{equation*}
  \left(
    \begin{array}{@{}l@{}}
      \mathbf{left} = T_{<\mathbf{l}_0}(\mathbf{l}) \wedge \mathbf{right} = T_{\geq \mathbf{l}_0}(\mathit{skip}_1(\mathbf{l})) \wedge\mbox{}
      \\
      (k \in \mathbf{left} \leftrightarrow k \in \mathbf{left}')
      \wedge (k \in \mathbf{right} \leftrightarrow k \in \mathbf{right}') \wedge \mbox{}
      \\
      \mathbf{res} = \mathbf{left}' \,.\, [\mathbf{l}_0] \,.\, \mathbf{right}'
    \end{array}
  \right)
  \to (k \in \mathbf{l} \leftrightarrow k \in \mathbf{res})
\end{equation*}
\OMIT{
    % the following is a weaker statement expressible w/o transducers
    % see slides
\begin{equation*}
  \left(
    \begin{array}{@{}l@{}}
        (k \in \mathbf{left} \rightarrow k \in \mathbf{l} \wedge k < \mathbf{l}) 
        \wedge (k \in \mathbf{right} \rightarrow k \in \mathbf{l} \wedge k >
        \mathbf{l}) \wedge\mbox{}
      \\
      (k \in \mathbf{left} \leftarrow k \in \mathbf{left}')
      \wedge (k \in \mathbf{right} \leftarrow k \in \mathbf{right}') \wedge \mbox{}
      \\
      \mathbf{res} = \mathbf{left}' \,.\, [\mathbf{l}_0] \,.\, \mathbf{right}'
    \end{array}
  \right)
  \to (k \in \mathbf{l} \leftarrow k \in \mathbf{res})
\end{equation*}
}
As one of the results of this paper, we prove that this final formula
is in a decidable logic. The formula can be rewritten to a disjunction
of straight-line formulas, and shown to be valid using the decision
procedure presented in Section~\ref{sec:sl}.

%%% Local Variables:
%%% mode: latex
%%% TeX-master: "main"
%%% End:

\section{Models}
\label{sec:model}

In this section, we will define our sequence constraint language, and prove
some basic results regarding various constraints in the language. The definition
is a natural generalization of string constraints (e.g. see
\cite{ganesh-word,diekert,LB16,popl19,wordequations}) by employing an alphabet
theory (a.k.a.\ element theory), as is done in symbolic automata and automata modulo theories
\cite{DV21,symbolic-power,veanes12}. 

For simplicity, our definitions will follow a 
model-theoretic approach. Let $\vocab$ be a vocabulary. We fix a 
$\vocab$-structure $\struct = (D; I)$, where $D$ can be a finite or an infinite 
set (i.e., the universe) and $I$
maps each function/relation symbol in $\vocab$ to a function/relation over
$D$. The elements of our sequences will range over $D$.
We assume that the quantifier-free theory $T_{\struct}$ over $\struct$ (including
equality) is decidable. Examples of such $T_{\struct}$ are abound from
SMT, e.g.,
%Linear Real Arithmetic and Linear Integer Arithmetic.
LRA and LIA.
We write $T$ instead of $T_{\struct}$, when $\struct$ is clear. 
Our quantifier-free formula will use \emph{uninterpreted
$T$-constants}
%\footnote{These could also equivalently be construed as 
%variables. We reserve the term ``variable'' to mean a sequence variable.} 
$a,b,c,\ldots$, and may also use variables $x,y,z,\ldots$.
(The distinction between uninterpreted constants and variables is made only
for the purpose of presentation of sequence constraints, as will be clear
shortly.)
%uninterpreted $T$-constants (i.e., they will get assigned $T$-elements). 
We use $\Constants$ to denote the set of all uninterpreted $T$-constants.
A formula $\varphi$ is satisfiable if there is an assignment that maps the
uninterpreted constants and variables to concrete values in $D$ such that the 
formula becomes true in $\struct$.
%For simplicity, we wil
%For the sake of clarity, we use the setting 
\OMIT{
For any tuple $\bar x$ of variables, we write 
$T_{\struct}(\bar x)$ to mean the set of formulas in $T_{\struct}$ that use only
variables in $\bar x$.
}

Next, we define how we lift $T$ to sequence constraints, using $T$ as the
\emph{alphabet theory} (a.k.a. \emph{element theory}). As in the case of
strings (over a finite alphabet), we use standard notation like $D^*$ to refer
to the set of all sequences over $D$. By default, elements of $D^*$ are written
as standard in mathematics, e.g., $7,8,100$, when $D = \Z$. Sometimes we will 
disambiguate them by using brackets, e.g., $(7,8,100)$ or $[7,8,100]$.
We will use the symbol $s$ (with/without
subscript) to refer to concrete sequences (i.e., a member of $D^*$). We will use
$\Xseq,\Yseq,\Zseq$ to refer to $T$-sequence variables.
Let $\Variables$ denote the set of all $T$-sequence variables, and
$\Gamma := \Constants \cup D$.
%\sideanthony{Do we use these $T$ variables?}
%\sideartur{I do not think so. And witout them we could use simple $x$ for the sequence variable,}
We will define 
constraint languages syntactically at the beginning, and will instantiate them to
specific sequence operations. The theory $\theory^*$ of $\theory$-sequences 
consists of the following constraints:
\[
    \varphi ::= R(\Xseq_1,\ldots,\Xseq_r) \ |\ \varphi \wedge \varphi
\]
where $R$ is an $r$-ary relation symbol. In our definition of each atom $R$
below, we will specify if an assignment $\mu$, which maps each $\Xseq_i$ to a
$T$-sequence and each uninterpreted constant to a $T$-element, satisfies $R$.
If $\mu$ satisfies all atoms, we say that $\mu$ is a \emph{solution}
and the \emph{satisfiability problem} is to decide whether there is a solution
for a given $\varphi$.
%The \Sem is not used anywhere later.
%We will denote by $\Sem{\varphi}$ the set of solutions $\mu$ of $\varphi$.
%The \emph{satisfiability problem} is to decide whether $\Sem{\varphi} \neq
%\emptyset$ for a given $\varphi$.
%%$\psi$ is a closed $T$-formula that may use $\Constants$ as uninterpreted
%constants, and
%$R \subseteq (\Gamma^*)^k$ is an $r$-ary relation over $\Gamma$-sequences.

A few remarks about the missing boolean operators in the constraint language
above are in order.
Disjunctions can be handled easily using the DPLL(T) framework (e.g.\ see
\cite{KS08}), so we have kept our theory conjunctive. As in the case of 
strings, negations are usually handled separately because they can sometimes 
(but not in all cases) be eliminated while preserving decidability. 
\OMIT{
We will next define different ways of instantiating relations over
sequences. The semantics is obvious: a solution $\mu$ is a function that maps
each uninterpreted $T$-constant $a$ to a $T$-element and each sequence 
variable $\Xseq$ to a $T$-sequence that is consistent with the interpretation of
uninterpreted $T$-constants, in such a way that $\varphi$ becomes true (here, 
$\wedge$
will be interpreted as an intersection in the standard way). To simplify our
presentation, 
we assume that \emph{each uninterpreted $T$-constant gets assigned different
$T$-elements by $\mu$}, i.e., $\mu(a) \neq \mu(b)$.
}

\paragraph{Equational constraints.} %The first kind of constraints is an 
%equational constraint}. 
A \defn{$T$-sequence equation} is of the form
\[
    L = R
\]
where each of $L$ and $R$ is a concatenation of concrete $T$-elements,
uninterpreted 
constants, and $T$-sequence variables. That is, if $\Theta := \Gamma \cup
\Variables$,
then $L,R \in \Theta^*$.
%, and $T$-variables. 
\OMIT{
For example, when $T$
refers to LRA, in the constraint
\[
    0.x = x.0
\]
the solutions are exactly $x \mapsto 0$. 
}
For example, in the equation
\[
    0. 1. \Xseq = \Xseq. 0.1
\]
the set of all solutions is of the form $\Xseq \mapsto (01)^*$. 
\OMIT{
To make this more formal, given a sequence equation $E$ with variables
$\Xseq_1,\ldots,\Xseq_r$, an assignment $\nu$ is a mapping from each
$T$-sequence variable $\Xseq$ to an element of $\Gamma^*$ and
each $a \in \Constants$ to an element of $D$. 
}
To make this more formal, we extend each assignment $\mu$
%so that each occurrence of $a$ in
%$\nu(\Xseq)$ gets assigned to $\nu(a)$.
% and each $T$-variable to $D$. 
to a homomorphism on $\Theta^*$. We write $\mu \models L = R$ if $\mu(L) = 
\mu(R)$. Notice that this definition is just direct extension of 
that of \emph{word equations} (e.g. see \cite{diekert}), i.e., when the domain
$D$ is finite.

In most cases the inequality constraints $L \neq R$ can be reduced to equality
in our case this requires also element constraints, described below.

\OMIT{
The generated
$r$-ary relation consists of all $r$-tuples of the form
$(\nu(\Xseq_1),\ldots,\nu(\Xseq_r))$, where $\nu$ is a solution to $E$.
}

\paragraph{Element constraints.} We allow $T$-formulas to constrain the
uninterpreted constants. More precisely, given a $T$-sentence (i.e., no free
variables) $\varphi$ that uses $\Constants$ as uninterpreted constants, we 
obtain a proposition $P$ (i.e., 0-ary relation) 
that $\mu \models P$ iff
%iff $|w_1| = \cdots = |w_r| = 1$ and 
$T \models_{\mu} \varphi$. 

Negations in the equational constraints can be removed just like in the
case of strings, i.e., by means of additional variables/constants and element
constraints. For example, $\Xseq \neq \Yseq$ can be replaced by $(\Xseq =
\Zseq a \Xseq' \wedge \Yseq = \Zseq b\Yseq' \wedge a \neq b) \vee
\Xseq = \Yseq a\Zseq \vee \Xseq a \Zseq = \Yseq$.
Notice that 
$a \neq b$ is a $T$-formula because we assume the equality symbol in $T$.

\paragraph{Regular constraints.}
Over strings, regular constraints are simply unary constraints $U( \Xseq )$, where 
$U$ is an automaton. The interpretation is $\Xseq$ is in the language of $U$. 
We define an analogue of regular constraints over sequences using 
\defn{parametric automata} \cite{FK20,FJL22,FL22}, which 
generalize both symbolic automata \cite{DV21,symbolic-power} and 
variable automata \cite{GKS10}.

A \defn{parametric automaton} (PA) over $T$ is
of the form $\Aut = (\params,\controls,\transrel,q_0,\finals)$, where
$\params$ is a finite set of parameters, $\controls$ is a finite set of
control states, $q_0 \in \controls$ is the initial state, $\finals \subseteq
\controls$ is the set of final states, and
$\transrel \subseteqf \controls \times T(\curr,\params) \times \controls$.
Here, \emph{parameters} are simply uninterpreted $T$-constants, i.e., 
$\params \subseteq \Constants$.
Formulas that appear in transitions in $\transrel$ will be referred to as
\emph{guards}, since they restrict which transitions are enabled at a given state.
Note that $\curr$ is an uninterpreted constant
that refers to the ``current'' position in the sequence.
The semantics is quite simply defined:
a sequence $(d_1, d_2, \ldots, d_n)$ is in the language of $\Aut$ under the assignment of parameters $\mu$,
written as $(d_1, \ldots, d_n) \in L_\mu(\Aut)$, when
there is a sequence of $\transrel$-transitions
\[
(q_0,\varphi_1(\curr, \params),q_1),
(q_1,\varphi_2(\curr, \params),q_2),
\ldots,
(q_{n-1},\varphi_n(\curr, \params),q_n),
\]
such that $q_n \in \finals$ and $T \models \varphi_i(d_i, \mu(\params))$.
Finally, for a regular constraint $\Aut(\Xseq)$ is satisfied by $\mu$, when
$\mu(\Xseq) \in L_\mu(\Aut)$.

Note, that it is possible to complement a PA $\Aut$,
one has to be careful with the semantics:
we treat $\Aut$ as a symbolic automaton,
which are closed under boolean operations~\cite{symbolic-power}.
So we are looking for $\mu$ such that $\mu(\Xseq) \in \overline{L_\mu(\Xseq)}$.
What we cannot do using the complementation,
is a universal quantification over the parameters;
note that already theory of strings with universal and existential quantifiers is undecidable.

We state next a
lemma showing that PAs using only ``local'' parameters, together
with equational constraints, can encode the constraint language that we have
defined so far.
\OMIT{
For any instantiation
$\mu: \params \to D$, we obtain a symbolic automaton $\Aut_{\mu}$ and so we can
use the definition of a symbolic automaton to define the semantics of a
run. Namely, $L(\Aut) = \bigcup_{\mu: \params \to D} L(\Aut_{\mu})$. A regular
constraint is then of the form $x \in L(\Aut)$, whose semantics is standard: 
we require $\nu(\Xseq) \in L(\Aut)$, for any given instantiation $\nu$ of $\Xseq$.
\sideanthony{Need to remove references to symbolic automata in the defn.}
}

\begin{lemma}
    Satisfiability of sequence constraints with equation, element, and regular 
    constraints can be reduced in polynomial-time to satisfiability of sequence
    constraints with equation and regular constraints (i.e., without element
    constraints). Furthermore, it can be assumed that no two regular constraints
    share any parameter.
    \label{lm:reduce}
\end{lemma}

\begin{proposition}
    Assume that $T$ is solvable in \np (resp.\ \pspace). Then, deciding
    nonemptiness of a parametric automaton over $T$ is in \np (resp.\ \pspace).
\end{proposition}
The proof is standard (e.g.\ see~\cite{FK20,FJL22,FL22}), and only sketched
here. The algorithm first nondeterministically guesses a simple path in
the automaton $\Aut$ from an initial state $q_0$ to some final state $q_F$. 
Let us say that the guards appearing in this path are $\psi_1(\curr,\params),
\ldots,\psi_k(\curr,\params)$. We need to check if this path is realizable by
checking $T$-satisfiability of 
\[
    \exists \params.\, \bigwedge_{i=1}^k \exists \curr.\,(\psi_i(\curr, \params)).
\]
It is easy to see that this is an \np (resp.\ NPSPACE =
\pspace) procedure.

\paragraph{Parametric transducers.}
We define a suitable extension of symbolic transducers over parameters following the definition from Veanes et al.~\cite{veanes12}.
A \emph{transducer constraint} is of the form $\Yseq = \Tra(\Xseq)$, for a
parametric transducer $\Tra$.
A \emph{parametric transducer} over $T$ is of the form 
$\Tra = (\params,\controls,\transrel,q_0,\finals)$, where $\params$,
$\controls$, $q_0$, $\finals$ are just like in parametric automata.
Unlike parametric automata, $\transrel$ is a finite set of tuples of the form
$(p,(\varphi,\Wseq),q)$, where $(p,\varphi,q)$ is a standard transition in
parametric automaton, and $\Wseq$ is a (possibly empty) sequence of $T$-terms 
over variable $curr$ and constants $\params$, e.g., $\Wseq = (curr+7,curr+2)$.
One can think of $\Wseq$ as the output produced by the transition.
Given an assignment $\mu$ of parameters and the sequence variables, the
constraint $\Yseq = \Tra(\Xseq)$ is satisfied when there is a sequence of
$\transrel$-transitions
\[
    (q_0,\varphi_1(\curr, \params),\Wseq_1,q_1),
    (q_1,\varphi_2(\curr, \params),\Wseq_2,q_2),
    \ldots
    (q_{n-1},\varphi_n(\curr, \params),\Wseq_n,q_n),
\]
such that $q_n \in \finals$ and $T \models \varphi_i(d_i, \mu(\params))$,
where $\mu(\Xseq) = (d_1,\ldots,d_n)$, and finally 
\[
    \mu(\Yseq) = \mu_1(\Wseq_1) \cdots \mu_n(\Wseq_n)
\]
where $\mu_i$ is $\mu$ but maps $\curr$ to $d_i$. The definition assumes that
$\mu_i$ is extended to terms and
concatenation thereof by homomorphism,
e.g., in LRA, if $\Wseq_1 = (\curr + 7,\curr + 2)$
and $\mu_1$ maps $\curr$ to 10, then $\Wseq_1$ will get mapped to $17,12$.
Given a set $S \subseteq D^*$ and an assignment $\mu$ (mapping the constants
to $D$), we define the \emph{pre-image}
$\Tra_{\mu}^{-1}(S)$ of $S$ under $\Tra$ with respect to $\mu$ as the set of 
sequences $\Wseq \in D^*$ such that $\Wseq' = \Tra(\Wseq)$ holds with respect to
$\mu$. 
%where $\params$ is a finite set of parameters, $Q$ is a finite set of states and $q_0$ is the initial state. $F \subseteq Q$ is the set of final states. $R$ is a set of rules $(p, \phi, f, q)$, where $p,q \in Q$. Similarly as with parametric automata, we call $\phi \in T(\curr, \params)$ a \emph{guard} that restricts transitions, and $f$ is a $T$-term over $\params$ and $\curr$ called \emph{output function}.  We can define the semantics of a run as follows:
%A run $w = a_1\dots a_n$ is a sequence of pairs $(q_0, t^1_1\dots t^n_1)\dots (q_n,t^1_n\dots t^n_n)$ where each $t^j_i$ is a $T$-term
%and $(t^1_j,\dots, t^n_i)$ is an output function. A run is accepting if $q_n \in F$.
%We call $w' = t^1_1\dots t^n_1\dots t^1_n\dots t^n_n$ the \emph{output} of the run.
%We use  $\mathcal{T}(Tr)$ to denote the relation containing all pairs $(w, w')$ such that $w'$ is an output of $Tr$.
%Again, we obtain a symbolic transducer for an instantiation for $\params$. We can then follow the definition of a symbolic automaton to define the semantics of a run.
%Namely, $\mathcal(Tr) = \bigcup_{\mu: \params \to D} \mathcal{T}(Tr_{\mu})$.
%Such a transducer can then be written as a function $f(x) = y$,  whose semantics is as above: 
%we require $\nu(x,y) \in \mathcal{T}(Tr)$, for any given instantiation $\mu$ of $x$.

\OMIT{
Finally, we define an extension of a rational transducer to sequences. To this
end, we simply just follow symbolic transducers \textit{a la.}\ Veanes since this seems to
serve most intents and purposes. \sideanthony{If required, we could also extend this
to allow parameters as well.}
}

%!TEX root = main.tex
\section{Solving Equational and Regular Constraints}
\label{sec:we}

Here we present results on solving equational constraints, together with
regular constraints, by a reduction to the string case, for which a wealth of
results are already available. In general, this reduction causes an exponential
blow-up in the resulting string constraint, which we show to be unavoidable in
general. That said, we also provide a more refined analysis in the case when the
underlying theory is LRA, where we can avoid this exponential blow-up.

%First, to streamline the algorithm
%!TEX root = main.tex
%\section{Preliminaries}
%\label{sec:prelim}
%Some of the constraints considered in the paper generalize word equations (with regular constraints):
%
\subsubsection*{Prelude: The case of strings}
\OMIT{
Given a finite alphabet $\Gamma$ and a finite set of variables $\mathcal V$ a system of word equations 
is given as a set of equations $\{L_i = R_i\}_{i \in I}$, where $L_i, R_i \in (\Gamma \cup \mathcal V)^*$
and constraints of the form $x \in L(A)$, where $x \in \mathcal V$ is a variable and $L(A)$
is language of an NFA $A$ (other specification of regular language can also be given).
Such system is satisfiable, when there is a substitution $\mu: \mathcal V \to \Gamma^*$
which turns each formal equality $U_i = V_i$ into a true equality of strings and for each regular constraint $x \in L(A)$
indeed $\mu(x) \in L(A)$.
}
We start with some known results about the case of strings.
The satisfiability of word equations with regular constraints is \pspace-complete~\cite{diekertfreegroups,wordeqgroups}. This upper bound can be
extended to full quantifier-free theory~\cite{buchi}.
When no regular constraints are given, the problem is only known to be \np-hard,
and it is widely believed to be in \np. 
In the absence of regular constraints, without loss of generality $\Gamma$ can be assumed to contain only letters from the equations;
this is not the case in presence of regular constraints.
The algorithm solving word equations~\cite{wordeqgroups} does not need an
explicit access to $\Gamma$:
it is enough to know whether there is a letter which labels a given set of transitions in the NFAs used in the regular constraints.
In principle, there could be exponentially many different (i.e., inducing different transitions in the NFAs) letters.
When oracle access to such alphabet is provided, the satisfiability can still be decided in \pspace:
while not explicitly claimed, this is exactly the scenario in \cite[Sect.~5.2]{wordeqgroups}

Other constraints are also considered for word equations;
perhaps the most widely known are the length constraints, which are of the form:
$\sum_{x \in \mathcal V} a_x \cdot |x| \leq c$,
where $\{a_x\}_{x \in \mathcal V}, c$ are integer constants
and $|x|$ denotes the length $|\mu(x)|$, with an obvious semantics.
It is an open problem, whether word equations with length constraints are 
decidable, see \cite{ganesh-word}.

\subsubsection*{Reduction to word equations}
We assume Lemma~\ref{lm:reduce},
i.e.\ that the parameters used for different automata-based constraints are pairwise different.
%i.e.\ for each constraint $\Xseq \in L(\Aut)$ we create a copy
%of $\Aut$ that uses fresh parameters.
%Note that this does not affect the satisfiability of constraints:
%if $\Aut'$ is a copy of $\Aut$ that uses $\params'$ instead of $\params$,
%then the defined languages are equal:
%$L(\Aut) = \bigcup_{\mu: \params \to D} L(\Aut_{\mu}) = \bigcup_{\mu': \params' \to D} L(\Aut'_{\mu'})$.
In particular, when looking for a satisfying assignment $\mu$
we can first fix assignment for $\params$
and then try to extend it to $\mathcal V$.
To avoid confusion, we call this partial assignment $\pi : \params \to D$.

Consider a set $\Phi$ of all atoms in all guards in the regular constraints
together with the set of formulas $\{x = c\}$ over all constants  
$c \in D$ that appear in all equational constraints and the negations of both types of formulas.
Fix an assignment $\pi : \params \to D$.
The type $\type_\pi(a)$ of $a$ (under assignment $\pi$)
is the set of formulas in $\Phi$ satisfied by $a$,
i.e.\ $\{ \phi \in \Phi \: : \:  \phi(\pi(\params), a) \text{ holds}\}$.
%for shortness we define $\sim_\pi$ whose equivalence classes are the values of the same type:
%\[
%a \sim_\pi b \iff \type_\pi(a) = \type_\pi(b)
%\]
Clearly there are at most exponentially many different types
(for a fixed $\pi$).
A type $t$ is realizable (for $\pi$) when $t = \type_\pi(a)$
and it is realized by $a$.
%equivalence classes.

If the constraints are satisfiable (for some parameters assignment $\pi$)
then they are satisfiable over a subset $D_\pi \subseteqf D$,
in the sense that we assign uniterpreted constants elements from $D_\pi$ and
$T$-sequence variables elements of $D_\pi^*$,
where $D_\pi$ is created by taking (arbitrarily) one element of a realizable type.
Note that for each constant $c$ in the equational constraints there is a formula ``$x = c$'' in $\Phi$,
in particular $\type_\pi(c)$ is realizable (only by $c$) and so $c \in D_\pi$.

\begin{lemma}
\label{lem:finite_domain_subset_satisfiability}
 Given a system of constraints and a parameter assignment $\pi$
 let $D_\pi \subseteq D$ be obtained by choosing (arbitrarily)
 for each realizable type a single element of this type.
 Then the set of constraints is satisfiable (for $\pi$) over $D$
 if and only if they are satisfiable (for $\pi$) over $D_\pi$.
 To be more precise, there is a letter-to-letter homomorphism
 $\psi : D^* \to D_\pi^*$ such that if $\mu$ is a solution of a system of constraints
 then $\psi \circ \mu$ is also a solution.
% 
% In particular, if $D_\pi, D_\pi'$ are two different such sets
% then the constraints are is equisatisfiable over them.
%
\end{lemma}

The proof can be found in the full version, its intuition is clear:
we map each letter $a \in D$ to the unique letter in $D_\pi$ of the same type.

%Note that Lemma~\ref{lem:finite_domain_subset_satisfiability}
%implies that it is enough to consider a solution over a domain that
%has at most one element per equivalence class.

Once the assignment is fixed (to $\pi$) and domain restricted to a finite set ($D_\pi$),
the equational and regular constraints reduce to word equations with regular constraints:
treat $D_\pi$ as a finite alphabet, for a parametric automaton
$\Aut = (\params,\controls,\transrel,q_0,\finals)$
create an NFA $\Aut' = (D_\pi,\controls,\transrel',q_0,\finals)$,
i.e.\ over the alphabet $D_\pi$, with the same set of states $Q$,
same starting state $q_0$ and accepting states $F$
and the relation defined as $(q,a,q') \in \transrel'$
if and only if there is $(q,\phi(\curr,\params),q') \in \transrel$
such that $\phi(a,\pi(\params))$ holds,
i.e.\ we can move from $q$ to $q'$ by $a$ in $\Aut'$ if and only if
we can make this move in $\Aut$ under assignment $\pi$.
Clearly, from the construction 

\begin{lemma}
\label{lem:parametric_to_regular}
	Given an assignment of parameters $\pi$ let $D_\mu$ be a set from Lemma~\ref{lem:finite_domain_subset_satisfiability},
	$\Aut$ be a parametric automaton and $\Aut'$ the automaton as constructed above.
	Then 
	\[
	L_\pi(\Aut) \cap D_\pi^* = L(\Aut') \enspace .
	\]
\end{lemma}

We can rewrite the parametric automata-constraints with regular constraints
and treat equational constraints as word equations (over the finite alphabet $D_\pi$).
%we only need to add simple constraints $x \in D_\pi$ for each $T$-variable $x$, to guarantee that it is substituted with a single element from $D_\pi$.
From Lemma~\ref{lem:finite_domain_subset_satisfiability}
and Lemma~\ref{lem:parametric_to_regular} it follows that the original constraints
have a solution for assignment $\pi$ if and only if the constructed system of constraints has a solution.
Therefore once the appropriate assignment $\pi$ is fixed,
the validity of constraints can be verified~\cite{wordeqgroups}.
It turns out that we do not need the actual $\pi$,
it is enough to know which types are realisable for it,
which translates to an exponential-size formula.
We will use letter $\tau$ to denote subset of $\Phi$;
the idea is that $\tau = \{\type_\pi(a) \: : \: a \in D\} \subseteq 2^\Phi$
and if different $\pi, \pi'$ give the same sets of realizable types,
then they both yield a satisfying assignment or both not.
Hence it is enough to focus on $\tau$ and not on actual $\pi$.

\begin{lemma}
\label{lem:parameters_type_verification}
Given a system of equational and regular constraints
we can non-deterministically reduce them to a formula of a form
\begin{equation}
\label{eq:parameters_type_verification}
\exists_{t \in \tau} a_t \in D .\, \exists \params \in D^+ .\, \bigwedge_{t \in \tau} \bigwedge_{\phi \in t}
\phi(\params, a_t) \enspace ,
\end{equation}
where $\tau \subseteq 2^\Phi$ is of at most exponential size,
and a system of word equations with regular constraints of linear size
and over an $|\tau|$-size alphabet, using auxiliary $\mathcal O(n |\tau|)$ space.
The solution of the latter word equations (for which also \eqref{eq:parameters_type_verification} holds)
are solutions of the original system, by appropriate identifications of symbols.
\end{lemma}

%The formula~\eqref{eq:parameters_type_verification} verifies that the guess of $\tau$ as the set of realisable types is correct;
%the proof can be fund in the appendix.
	\begin{proof}
		We guess the set $\tau$ of types of the assignment of parameters $\pi$,
		i.e.\ $\tau = \{\type_\pi(a) \: : \: a \in D\}$
		such that there is an assignment $\mu$ extending $\pi$;
		note that as $\Phi$ has linearly many atoms and $\tau \subseteq 2^\Phi$,
		then $|\tau|$ may be of exponential size, in general.
		%Note that we do not need to guess the actual values of $\pi(\params)$,
		%it is enough to guess the types in $\type(\pi)$.
		The~\eqref{eq:parameters_type_verification} verifies the guess:
		%i.e.\ whether indeed there is $\pi$ such that the $\type(\pi) = \tau$:
		we validate whether there are values of $\params$ such that
		for each type $t \in \tau$ there is a value $a$ such that $\type_\pi(a) = t$.

		Let $D_\pi$ be a set having one symbol per every type in $\tau$,
		as in Lemma~\ref{lem:finite_domain_subset_satisfiability};
		note that this includes all constants in the equational constraints.
		The algorithm will not have access to particular values,
		instead we store each $t \in \tau$, say as a bitvector describing which atoms in $\Phi$ this letter satisfies.
		In particular, $|D_\pi| = |\tau|$ and it is at most exponential. % in the size of the original instance.
		In the following we will consider only solutions over $D_\pi$.
		
		For each $a \in D_\pi$ we can validate,
		which transitions in $\Aut$ it can take:
		the transition is labelled by a guard which is a conjunction of atoms from $\Phi$ and
		either each such atom is in $\type_\pi(a)$ or not.
		Hence we can treat $\Aut$ as an NFA for $D_\pi$.
		We do not need to construct nor store it, we can use $\Aut$:
		when we want to make a transition by $\phi(\params,a)$ we look up, whether each atom
		of $\phi$ is in $\type_\pi(a)$ or not.
		Similarly, the constraint $\Aut(\Xseq)$ is restricted to
		$\Xseq \in L_\pi(\Aut)$ and for $\Xseq \in D_\pi^*$ this is a usual regular constraint.
		
		We treat equational constraints as word equations over alphabet $D_\pi$.
%		what we need is a simple additional constraint that $T$-variables are indeed substituted by single letters.
%		To this end we add a simple regular constraints $x \in D_\pi$ for them.

		%Note that it could be that for actual $\pi$ we have $\tau \subsetneq \type(\pi)$;
		%this is not a problem: given $\tau$ and $t \in \tau$
		%we rely on existence of a letter $a$ such that $\type(a) = t$.
		%We do not use that fact that $t' \notin \tau$ for some $t'$.

		Concerning the correctness of the reduction:
		if the system of word equations (with regular constraints) is satisfiable
		and the formula~\eqref{eq:parameters_type_verification} is also satisfiable,
		then there is a satisfying assignment $\mu$ over $D_\pi$ and $D_\pi^*$
		in particular, there is an assignment of parameters for which there are letters
		of the given types (note that in principle it could be that $\mu$ induces more types,
		i.e.\ there is a value $a$ such that $\type_\mu(a) \notin \tau$ and so it is not represented in $D_\pi$,
		but this is fine: enlarging the alphabet cannot invalidate a solution),
		i.e.\ the transitions for $a_t$ in the automata after the reduction are the same
		as in the corresponding parametric automata for the assignment $\pi$,
		this is guaranteed by the satisfiability of~\eqref{eq:parameters_type_verification}
		and the way we construct the instance, see Lemma~\ref{lem:parametric_to_regular}.

		On the other hand, when there is a solution of the input constraints,
		there is one for some assignment of parameters $\pi$.
		Hence, by Lemma~\ref{lem:finite_domain_subset_satisfiability},
		there is a solution over $D_\pi$. 
		%which has one letter %per each type $t \in \type(\pi)$.
		The algorithm guesses $\tau = \{\type_\pi(a) \: : \: a\in D\}$ and
		\eqref{eq:parameters_type_verification} is true for it.
		Then by Lemma~\ref{lem:finite_domain_subset_satisfiability}
		there is a solution over $D_\pi$ as constructed in the reduction
		and by Lemma~\ref{lem:parametric_to_regular} the regular constraints define
		the same subsets of $D_\pi^*$ both when interpreted as parametric automata
		and NFAs.
		\qed
	\end{proof}

\begin{theorem}
\label{thm:sequence_constraints_expspace}
If theory $T$ is in \pspace then sequence constraints are in \expspace.

If $\tau$ is polynomial size and the formula~\eqref{eq:parameters_type_verification} can be verified in \pspace,
then sequence constraints can be verified in \pspace.
\end{theorem}

One of the difficulties in deciding sequence constraints using the word equations approach
is the size of set of realizable types $\tau$,
which could be exponential.
For some concrete theories it is known to be smaller and thus a lower upper bound on complexity follows.
For instance, it is easy to show that for %Linear Real Arithmetics
LRA there are linearly many realizable types,
which implies a \pspace upper bound.

\begin{corollary}\label{cor:LRA}
Sequence constraints for Linear Real Arithmetic are in \pspace.
\end{corollary}

In general, the \expspace upper bound from Theorem~\ref{thm:sequence_constraints_expspace}
cannot be improved, as even non-emptiness of intersection of parametric automata
is \expspace-complete for some theories decidable in \pspace.
This is in contrast to the case of symbolic automata,
for which the non-emptiness of intersection (for a theory $\theory$ decidable in \pspace)
is in \pspace.
This shows the importance of parameters in our lower bound proof.

\begin{theorem}\label{thm:we-expspace}
There are theories with existential fragment decidable in \pspace and whose non-emptiness of intersection of parametric automata
is \expspace-complete.
\end{theorem}

When no regular constraints are allowed,
we can solve the equational and element constraints in \pspace (note that we do not use Lemma~\ref{lm:reduce}).

\begin{theorem}
\label{thm:no_reg_constraints}
For a theory $\theory$ decidable in \pspace,
the element and equational constraints (so no regular constraints) can be decided in \pspace.
\end{theorem}

%!TEX root = cav23.tex
\section{Algorithm for straight-line formulas}
\label{sec:sl}

\OMIT{
Many  programs which manipulate lists, arrays, or streams use functions like map, split, filter etc.
In order to model such operations in our theory of sequences, one can employ parametric transducers.
}
It is known that adding finite transducers into word equations results in an
undecidable model (e.g.\ see \cite{LB16}). Therefore, we extend the 
\emph{straight-line restriction}~\cite{LB16,popl19} to sequences, and show that
it suffices to recover decidability for equational constraints, together with
regular and transducer constraints.
%Simply adding those transducers leads to an undecidable theory.
\OMIT{
However, instead of settling with an incomplete solver, we follow an approach inspired by solvers for the theory of strings~\cite{popl19}, and restrict ourselves to a fragment known as ``straight-line'' (SL), i.e.,
impose syntactic restrictions on our constraint language.
}
%Furthermore, ensuring that our input language satisfies two semantic conditions introduced later in this section, one can add transducers to the theory of sequences while remaining decidable.
In fact, we will show that deciding problems in the straight-line fragment is
solvable in doubly exponential time and is \expspace-hard, if $T$ is solvable 
in \pspace.
%The fragment we are interested in is known as the straight-line fragment (SL) which imposes syntactic restrictions on our constraint language\cite{LB16,popl18}. 
It has been observed that the straight-line fragment for the theory of strings already covers many interesting benchmarks~\cite{LB16,popl19}, and similarly many properties of sequence-manipulating programs can be proven using the fragment, including the QuickSort example from Section~\ref{sec:mot_ex} and other benchmarks shown in Section~\ref{sec:impl}.

\subsubsection*{The Straight-line Fragment SL}

%First, we define two semantic conditions.
We start by defining recognizable formulas over sequences, followed by the syntactic and semantic restrictions on our constraint language.
This definition follows closely the definition of recognizable
relations over finite alphabets, except that we replace finite automata with
parametric automata.
\begin{definition}[Recognizable formula] \label{def:rec-rel} A
    formula $R(\Xseq_1,\ldots,\Xseq_r)$ is
  \emph{recognizable} if it is equivalent to a positive Boolean
    combination of regular constraints.  
    \OMIT{
    A \emph{representation} of a
  recognizable relation $R$ is a family
  $(\Aut^{(i)}_1, \dots,\Aut^{(i)}_r )_{1 \leq i \leq n}$ of 
  parametric automata~$\Aut^{(i)}_j$, such that
  $(\Xseq_1,\dots, \Xseq_r) \in R$ iff there is $i$ with
  $1 \leq i \leq n$ and $\Aut^{(i)}_j(\Xseq_j)$ for
  $j = 1, \ldots, r$. The tuples $(\Aut^{(i)}_1, \dots,\Aut^{(i)}_r )$
  are called the \emph{disjuncts} of the representation, and the
  parametric automata $\Aut^{(i)}_j$ are called the \emph{atoms}.
  }
\end{definition}
Note that this is simply a generalization of regular constraints to multiple
variables, i.e., 1-ary recognizable formula can be turned into a regular
constraint, which is closed under intersection and union.

To define the straight-line fragment, we use the approach of \cite{popl19};
that is, the fragment is defined in terms of ``feasibility of a symbolic 
execution''. Here, a \emph{symbolic execution} is just a sequence of assignments and 
assertions, whereas the \emph{feasibility} problem amounts to deciding whether
there are concrete values of the variables so that the symbolic execution can be
run and none of the assertions are violated. We now make this intuition
formal. A symbolic execution is syntactically generated by the following
grammar:
%The language of straight-line formulas SL is defined by the following grammar:
\begin{align}
	&S ~~::=~~ \Yseq :=f(\Xseq_1,\dots, \Xseq_k, \params) \;|\; \textbf{assert}(R(\Xseq_1,\dots,\Xseq_r)) \;|\; \textbf{assert}(\phi) \;|\;S;S  &
\end{align}
where $f: (D^*)^k \times D^{|\params|} \to D$ is a function, $R$ are
recognizable formulas,
%represented by a collection of tuples of
%parametric automata, 
and $\phi$ are element constraints. 
\OMIT{
Assignments
to sequence variables in an SL formula have to be sorted: in a formula
$S_1; \ldots; S_{i-1}; \Yseq := f(\ldots); S_{i+1}; \ldots; S_n$, the
variable~$\Yseq$ can occur in the suffix~$S_{i+1}; \ldots; S_n$, but
must not occur in the prefix $S_1; \ldots; S_{i-1}$.
}

The symbolic execution $S$ can be turned into a sequence constraint as follows.
Firstly, we can turn $S$ into the standard \emph{Static Single Assignment (SSA)}
form by means of introducing new variables on the left-hand-side of an 
assignment. For example, $\Yseq := f(\Xseq); \Yseq := g(\Zseq)$ becomes
$\Yseq := f(\Xseq_1); \Yseq' := g(\Zseq)$. Then, in the resulting constraint,
each
variable appears \emph{at most once} on the left-hand-side of an assignment. That way,
we can simply replace each assignment symbol $:=$ with an equality symbol $=$.
We then treat each sequential composition as the conjunction operator $\wedge$
and assertion as a conjunct. Note that individual assertions are already sequence
constraints.
Next, we define how an interpretation $\mu$ satisfies the constraint 
$\Yseq = f(\Xseq_1,\ldots,\Xseq_r,\params)$:
\[
    \mu \models \Yseq = f(\Xseq_1,\ldots,\Xseq_r,\params) \quad \text{iff} \quad
    \mu(\Yseq) = f(\mu(\Xseq_1),\ldots,\mu(\Xseq_r),\mu(\params)).
\]
Note that '=' on the l.h.s. is syntactic, while the '=' on the r.h.s. is in the
metalanguage. The definition of the semantics of the language is now inherited
from Section \ref{sec:model}.

\OMIT{
For an assignment $\mu$ and every function assignment $\Yseq :=f(\Xseq_1,\dots, \Xseq_k, \params)$ we obtain a sequence constraint which is an equality $\Yseq :=f(\Xseq_1,\dots, \Xseq_k)$ where $\mu \models \Yseq :=f(\Xseq_1,\dots, \Xseq_k)$.
Every $\textbf{assert}(\phi)$ is an element constraint with $\mu \models \phi$. Finally, for $R(\Xseq_1,\dots,\Xseq_r))$ we obtain a Boolean combination of regular constraints. We say that $S$ is satisfiable iff the conjunction of all sequence constraints is satisfiable.
}

%The first semantic condition regulates the formulas~$R$
%that can be used in SL:
\OMIT{
\begin{description}
    \item[(RegMonDec)]
%\begin{definition}[RegMonDec] \label{def:reg-mon-dec}
	For each $\textbf{assert}(R(\Xseq_1,\dots, \Xseq_r))$ in $S$, it is possible
    to compute an equivalent formula to $R$ that is recognizable.
    %formula such that a representation of 
    %$R$ 
    %in terms of Definition \ref{def:rec-rel}, is effectively computable.
\end{definition}
In particular, \textbf{RegMonDec} is satisfied when 
        $R(\Xseq_1, \dots, \Xseq_r)$ is
recognizable, and for $r = 1$ it may be given as an regular constraint
$\mathcal{A}(\Xseq_1)$.
}
In addition to the syntactic restrictions, we also need a semantic
condition: in our language, we only permit functions~$f$ such that
the pre-image of each regular constraint under $f$ is effectively a
recognizable formula:
        \begin{description}
            \item[(RegInvRel)] A function~$f$ is permitted if for each regular
                constraint $\Aut( \Yseq)$, it is possible to compute
                a recognizable formula that is equivalent to the formula
                $\exists \Yseq:
                \Aut(\Yseq) \wedge \Yseq = f(\Xseq_1,\ldots,\Xseq_r,\params)$.
        \end{description}
%The second semantic condition concerns the pre-images of sequence functions:

%\begin{theorem} \label{theorem:decidable}
%	The path feasibility problem is decidable for symbolic executions satisfying Definition \ref{def:reg-inv-rel} and Definition \ref{def:reg-mon-dec}.
%\end{theorem}
%Problems inside the constraint language SL can be translated to a constraint in the theory of sequences which is a conjunction of function assignments and assertions.
        Two functions satisfying (RegInvRel) 
%
        %whose pre-images can be computed effectively 
        are the concatenation function $\Xseq := \Yseq.\Zseq$ (here $\Yseq$
        could be the same as $\Zseq$) and
        parametric transducers~$\Yseq := \Tra(\Xseq)$.
        We will only use these two functions in the paper, but the result is
        generalizable to other functions.
%The complexity result for our algorithm is dependent on the complexity of the pre-image computation, hence for our complexity result we assume that only concatenation and parametric transducers are used as functions.
%The following proposition gives the construction of a pre-image of a
%regular constraint that is done uniformly across the uninterpreted 
%$T$-constants.

\begin{proposition}\label{prop:concat}
	Given a regular constraint $\Aut( \Yseq )$ and a
    constraint $\Yseq = \Xseq.\Zseq$, we can compute a recognizable
    formula $\psi(\Xseq,\Zseq)$ equivalent to
	%$\Aut'( \Xseq )$ that is equivalent to 
    $\exists \Yseq: \Aut(\Yseq) \wedge
	\Yseq = \Xseq.\Zseq$. Furthermore, this can be achieved in polynomial 
    time.
\end{proposition}
The proof of this proposition is exactly the same as in the case of strings,
e.g., see \cite{popl19,LB16}.
\begin{proposition}\label{prop:pre-image}
	Given a regular constraint $\Aut( \Yseq )$ and a parametric transducer 
	constraint $\Yseq = \Tra(\Xseq)$, we can compute a regular constraint
	$\Aut'( \Xseq )$ that is equivalent to $\exists \Yseq: \Aut(\Yseq) \wedge
	\Yseq = \Tra(\Xseq)$. This can be achieved in exponential time.
	%regular constraint $
	%Computing the pre-image of a parametric transducer is a procedure in \exptime.
\end{proposition}
The construction in Proposition \ref{prop:pre-image} is essentially the same 
as the pre-image
computation of a symbolic automaton under a symbolic transducer \cite{veanes12}.
The complexity is exponential in the maximum number of output symbols of a 
single transition (i.e. the maximum length of $\Wseq$ in the transducer), which
is in practice a small natural number. \ifextendedtr For completeness, we have added the proof
in the appendix.\fi

%\begin{proposition}
%  The SL language satisfies the two semantic conditions RegMonDec and
%  RegInvRel.
%\end{proposition}

%\begin{proof}
%	The pre-image of parametric transducer can be reduced to the pre-image computation of symbolic finite transducer which can be done in exponential time. \cite{???}.
%\end{proof}

The following is our main theorem on the SL fragment with equational
constraints, regular constraints, and transducers.
\begin{theorem} 
	If $T$ is solvable in \pspace, then the SL fragment with concatenation and
    parametric transducers over $T$ is in 
    2-{\exptime} and is \expspace-hard.
    \label{thm:sl-expspace}
\end{theorem}
\begin{proof}
	We give a decision procedure\ifextendedtr; complexity analysis is given in the appendix\fi.
    We assume
    that $S$ is already in SSA (i.e. each variable appears at most once on the
    left-hand side).
    Let us assume that $S$ is of the form $S';\Yseq:= f(\Xseq_1,...\Xseq_r)$,
    for some symbolic execution $S'$.
    %be the last\footnote{This notion makes sense since ';' is a sequential
    %composition.} assignment in $S$. 
    Without loss of generality, we may assume
    that each recognizable constraint is of the form $\Aut(
    \Xseq)$. This is no limitation: (1) since each $R$ in the
    assertion is a recognizable formula, we simply have to
    ``guess'' one of the implicants for each $R$, and (2) $\textbf{assert}(\psi_1
    \wedge \psi_2)$ is equivalent to $\textbf{assert}(\psi_1);
    \textbf{assert}(\psi_2)$.

    Assume now that $\{\Aut_1(\Yseq),\ldots,\Aut_m(\Yseq)\}$ are all the regular
    constraints on $\Yseq$ in $S$. By our assumption, it is possible to compute
    a recognizable formula equivalent to
    \[
        \psi(\Xseq_1,\ldots,\Xseq_r) := \exists \Yseq: 
            \bigwedge_{i=1}^m \Aut_i(\Yseq) \wedge \Yseq =
                f(\Xseq_1,\ldots,\Xseq_r).
    \]
    There are two ways to see this. The first way is that regular constraints
    are closed under intersection. This is in general computationally quite
    expensive because of a product automata construction before applying the
    pre-image computation. A better way to do this is to observe that 
    $\psi$ is
    equivalent to the conjunction of $\psi_i$'s over $i=1,\ldots,m$, where 
    \[
        \psi_i := \exists \Yseq: \Aut_i(\Yseq) \wedge \Yseq =
        f(\Xseq_1,\ldots,\Xseq_r).
    \]
    By our semantic condition, we can compute recognizable formulas
    $\psi_i',\ldots,\psi_m'$ equivalent to $\psi_1,\ldots,\psi_m$ respectively.
    Therefore, we simply replace $S$ by 
    \[
        S';\textbf{assert}(\psi_1');\cdots;\textbf{assert}(\psi_m'),
    \]
    in which every occurrence of $\Yseq$ has been completely eliminated.
    Applying the above variable elimination iteratively, we obtain a
    conjunction of regular constraints. We now end up with a conjunction of
    regular constraints and element constraints, which as we saw
    from Section \ref{sec:we} is decidable.
	%
    %We analyze the complexity of our algorithm in the case when each $f$ is
    %either a concatenation or a parametric transducer. 
	%
    %and 
	\qed
\end{proof}

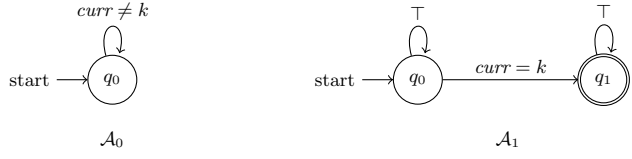
\begin{figure}[tb]
	\centering
	\begin{tikzpicture}[scale=1, every node/.style={scale=0.8}]

\node[state, initial] (1) {$q_0$};
\draw 
(1) edge[->,loop above] node{$\curr \neq k $} (1);
\node[ below of=1] {$\mathcal{A}_0$};

%	\node[state, initial, right=of 1, xshift = 1cm] (2) {$q_0$};
%	\node[state, right=of 2, xshift = 1cm] (3) {$q_1$};
%	\draw 
%	(2) edge[->, above] node{$\top \backslash()$} (3)
%	(3) edge[->,loop above] node{$\curr < p \backslash()$} (3)
%	(3) edge[->,loop below] node{$\curr \geq p \backslash(\curr)$} (3);
%	
%	\node[state, initial, right=of 3, xshift = 1cm] (4) {$q_0$};
%	\node[state, right=of 4, xshift = 1cm] (5) {$q_1$};
%	\draw 
%	(4) edge[->, above] node{$\curr = p$} (5)
%	(5) edge[->,loop above] node{$\top$} (5);
	
	\node[state, initial, right=of 1, xshift = 3cm] (6) {$q_0$};
	\node[state, right=of 6, xshift = 1cm, accepting] (7) {$q_1$};
	\draw 
	(6) edge[->,loop above] node{$\top$} (6)
	(6) edge[->, above] node{$\curr = k$} (7)
	(7) edge[->,loop above] node{$\top$} (7);
	\node[ below of=6, xshift= 1.5cm] {$\mathcal{A}_1$};

\end{tikzpicture}

	\caption{$\Aut_0$ accepts all words not containing $k$ and $\Aut_1$ accepts all words containing $k$.}	\label{fig:product-automaton}
\end{figure}

\begin{example}
	We consider the example from Section~\ref{sec:mot_ex} where a weaker form of the permutation property is shown for QuickSort.
	The formula that has to be proven is a disjunction of straight-line formulas and in the following we execute our procedure only on one disjunct without redundant formulas:
	%	\footnote{$\circ$ is the concatenation of two sequences and can be seen written as a transducer T(x,y)}
    \[\textbf{assert}(\Aut_0(\mathbf{left}'));\textbf{assert}(\Aut_0(\mathbf{right}'));
    \mathbf{res} = \mathbf{left}' \,.\, [\mathbf{l}_0] \,.\, \mathbf{right}'; 
    \textbf{assert}(\Aut_1(\mathbf{res})) \]

	We model $L(\Aut_1)$ as the language which accepts all words which contain one letter equal to $k$ and $L(\Aut_0)$ as the language which accepts only words not containing $k$, where $k$ is an uninterpreted constant, so a single element.
	See	Figure \ref{fig:product-automaton}.
	We begin by removing the operation  $\mathbf{res} = \mathbf{left}' \,.\,
    [\mathbf{l}_0] \,.\, \mathbf{right}'$. The product automaton for all
    assertions that contain $\mathbf{res}$ is just $\Aut_1$. Hence, we can
    remove the assertion $\textbf{assert}( \Aut_1(\mathbf{res}))$. 
	 The concatenation function $.$ satisfies \textbf{RegInvRel} and the
     pre-image~$g$ can be represented by 
     \[
         \bigvee_{0\leq i,j\leq 1} \Aut_1^{q_0,\{q_i\}}(\mathbf{left'}) \wedge 
            \Aut_1^{q_i, \{q_j\}}([\mathbf{l}_0]) \wedge 
            \Aut_1^{q_j, \{q_1\}}(\mathbf{right}'),
    \]
    where $\Aut_i^{p,\finals'}$ is $\Aut_i$ with start state set to $p$ and finals to $\finals'$.
	 
	In the next step, the assertion $g$ is added to the program and all assertions containing $\mathbf{res}$ and the concatenation function are removed.
	\begin{gather*}
        \textbf{assert}(\Aut_0(\mathbf{left}'));\textbf{assert}(\Aut(\mathbf{right}'));\textbf{assert}(g(\mathbf{left}',[\mathbf{l}_0], \mathbf{right}'))
	\end{gather*}	
	From here, we pick a tuple from $g$,  lets say $i = j = 1$, and obtain	
	\begin{gather*}
        \textbf{assert}(\Aut_0(\mathbf{left}'));\textbf{assert}(\Aut_0(\mathbf{right}'));
        \textbf{assert}(\mathbf{left}' \in \Aut_1^{q_0,\{q_1\}});\\
        \textbf{assert}([\mathbf{l}_0] \in
        \Aut_1^{q_1,\{q_1\}});\textbf{assert}(\mathbf{right}' \in
        \Aut_1^{q_1,\{q_1\}})
	\end{gather*}
	
    Finally, the product automata $\Aut_0 \times \Aut_1^{q_0,\{q_1\}}$ and
    $\Aut_0 \times \Aut_1^{q_0,\{q_1\}}$ are computed for the variables
    $\mathbf{left}', \mathbf{right}'$ and a non-emptiness check over the product
    automata and the automaton for $[\mathbf{l}_0]$ is done.
	The procedure will find no combination of paths for each automaton which can
    be satisfied, since $\mathbf{left}'$ is forced to accept no words containing
    $k$ by $\Aut_0$ and only accepts by reading a $k$ from $\Aut_1^{q_0,\{q_1\}}$.
    Next, the procedure needs to exhaust all tuples from $(\Aut_1^{q_0,\{q_i\}},
    \Aut_1^{q_i, \{q_j\}}, \Aut_1^{q_j, \{q_1\}})_{0 \leq i,j \leq 1}$ before it is proven that this disjunct is unsatisfiable.
	%	\begin{align*}
		%		\exists p & [([\exists curr \; curr = p] \lor [(\exists curr \; curr = p)\land (\exists curr \; curr = p \land curr \leq p)] ) \\ 
		%		\land & ([\exists curr \; curr < 10 \land curr <p ] \lor [(\exists curr \; curr <p) \land (\exists curr \; cur <10) ])]
		%	\end{align*} 
	
\end{example}

\section{Extensions and undecidability}
\label{sec:ext}
\subsubsection*{Length constraints}

We consider the extension of our model by allowing \emph{length-constraints} on the sequence variables:
for each sequence variable $\Xseq$ we consider the associated length variable $\ell_\Xseq$,
let the set of length variables be $\mathcal L =\{ \ell_\Xseq \: : \: \Xseq \in \mathcal V\}$,
we extend $\mu$ to $\mathcal L$, it assigns natural numbers to them.
The length constraints are of the form $\sum_{\Xseq} a_\Xseq \ell_\Xseq ? 0$,
where $? \in \{<, \leq,  =, \neq , \geq , >\}$ and each $a_\Xseq$ is an integer constant,
i.e., linear arithmetic formulas on the length-variables.
The semantics is natural: we require that $|\mu(\Xseq)| = \mu (\ell_\Xseq)$ (the assigned values are the true lengths of sequences)
and that $\mu(\mathcal L)$ satisfies each length constraint.

There is, however, another possible extensions:
if we the theory $T_\struct$ is the Presburger arithmetic, then the parameter automata
could use the values $\ell_\Xseq$.
We first deal with a more generic, though restricted case, when this is not allowed:
then all reductions from Section~\ref{sec:we} generalize
and we can reduce to the word equations with regular and length constraints.
However, the decidability status of this problem is unknown.
When we consider Presburger arithmetic and allow the automata to employ the length variables,
then it turns out that we can interpret the formula~\eqref{eq:parameters_type_verification}
as a collection of length constraints, and again we reduce to word equations with regular and length constraints.

\paragraph{Automata oblivious of lengths}
We first consider the setting, in which the length variables $\mathcal L$ can only be used in length constraints.
It is routine to verify that the reduction from Section~\ref{sec:we} generalize to the case of length constraints:
it is possible to first fix $\mu$ for parameters, calling it again $\pi$.
Then Lemma~\ref{lem:finite_domain_subset_satisfiability}
shows that each solution $\mu$ can be mapped by a letter-to-letter homomorphism to a finite alphabet $D_\pi$,
and this mapping preserves the satisfiability/unsatisfiability of length constraints,
so Lemma~\ref{lem:finite_domain_subset_satisfiability} still holds when also length constraints are allowed.
Similarly, Lemma~\ref{lem:parametric_to_regular} is also not affected by the length constraints
and finally Lemma~\ref{lem:parameters_type_verification} deals with regular and equational constraints,
ignoring the other possible constraints and the length of substitutions for variables are the same.
Hence it holds also when the length constraints are allowed
then the resulting word equations use regular and length constraints.

Unfortunately, the decidability of word equations with linear length constraints (even without regular constraints)
is a notorious open problem.
Thus instead of decidability, we get Turing-equivalent problems.

\begin{theorem}
\label{thm:length_constraints_simple}
Deciding regular, equational and length constraints for $T$-sequences of a decidable theory $\theory$
is Turing-equivalent to word equations with regular and length constraints.
\end{theorem}

\paragraph{Automata aware of the sequence lengths}
We now consider the case when the underlying theory $T_\struct$ is the Presburger arithmetic,
i.e.\ $\struct$ is the natural numbers and we can use addition, constants $0, 1$ and comparisons (and variables).
The additional functionality of the parametric automaton $\Aut$ is that
$\Delta \subseteqf \controls \times T(\curr,\params, \mathcal L) \times \controls$,
i.e.\ the guards can also use the length variables;
the semantics is extended in the natural way.

Then the type $\type_{\pi}(a)$ of $a \in \mathbb N$
now depends on $\mu$ values on $\params$ and $\mathcal L$,
hence we denote by $\pi$ the restriction of $\mu$ to $\params \cup \mathcal L$.
Then Lemma~\ref{lem:finite_domain_subset_satisfiability},
\ref{lem:parametric_to_regular} still hold, when we fix $\pi$.
Similarly, Lemma~\ref{lem:parameters_type_verification} holds,
but the analogue of \eqref{eq:parameters_type_verification} now uses also the length variables,
which are also used in the length constraints.
Such a formula can be seen as a collection of length constraints for original length variables $\mathcal L$
as well as length variables $\params \cup \{a_t \: : \: t \in \tau\}$.
Hence we validate this formula as part of the word equations with length constraints.
Note that $a_t$ has two roles: as a letter in $D_{\pi}$ and as a length variable.
However, the connection is encoded in the formula from the reduction (analogue of~\eqref{eq:parameters_type_verification})
and we can use two different sets of symbols.

\begin{theorem}\label{thm:ext-length}
Deciding conjunction of regular, equational and length constraints for sequences of natural numbers with Presburger arithmetic,
where the regular constraints can use length variables,
is Turing-equivalent to word equations with regular and (up to exponentially many) length constraints.
\end{theorem}

\subsubsection*{Undecidability of register automata constraints}
\label{subsec:register_automata}
One could use more powerful automata for regular constraints;
one such popular model are register automata;
informally, such automaton has $k$ registers $r_1, \ldots, r_k$
and its transition depends on state and a value of formula using the registers and $\curr$: the read value~\cite{FJL22};
note that the registers can be updated: to $\curr$ or to one of register's values;
this is specified in the transition.
In ``classic'' register automata guards can only use equality and inequality between registers and $\curr$;
in SRA model more powerful atoms are allowed.
We show that sequence constraints and register automata constraints
(which use quantifier-free formulas with equality and inequality as only atoms, i.e.\ do not employ the SRA extension)
lead to undecidability (over infinite domain $D$).

\begin{theorem}
Satisfiability of equational constraints and register automata constraints, which use equality and inequality only,
over infinite domain, is undecidable.
\end{theorem}

%The proof can be found in the appendix.

\section{Implementations, Optimizations and Benchmarks}
\label{sec:impl}

\subsubsection{Implementation}

We have implemented our decision procedure for problems in the constraint language SL for the theory of sequences in a new tool {\seqSolver} (Sequence Constraint Solver) on top of the SMT solver Princess \cite{princess}.
We extend a publicly available library for symbolic automata and transducers\cite{svpalib} to parametric automata and transducers by connecting them to the uninterpreted constants in our theory of sequences.
Our tool supports symbolic transducers, concatenation of sequences and reversing of sequences. Any additional function which satisfies \textbf{RegInvRel} such as a replace function which replaces only the first and leftmost longest match can be added in the future.

Our algorithm is an adaption of the tool OSTRICH \cite{popl19} and closely
follows the proof of Theorem \ref{thm:sl-expspace}.
To summarize the procedure, a depth-first search is employed to remove all functions in the given input and splitting on the pre-images of those functions.
When removing a function, new assertions are added to the pre-image constraints.
After all functions have been removed and only assertions are left a nonemptiness check is called over all parametric automata which encoded the assertions.
If the check is successful a corresponding model can be constructed, otherwise the procedure computes a conflict set and back-jumps to the last split in the depth search.\footnote{For a more detailed write-up of the depth-first search algorithm see OSTRICH\cite{popl19} Algorithm 1.}

\subsubsection{Benchmarks}
We have performed experiments on two benchmark suites.
The first one concerns itself with the verification of properties for programs manipulating sequences.
The second benchmark suite compares our tool against an algorithm using symbolic register automata \cite{svpalib} on decision procedures of regular expressions with back-references such as emptiness, equivalence and inclusion.

Both benchmark suites require universal quantification over the parameters;
there are existing methods for eliminating these universal quantifiers,
one such class are the \emph{semantically deterministic} (SD)~\cite{semantically_deterministic_PA} PAs;
despite its name, being SD is algorithmically checkable.
Most of considered the PAs are SD,
in particular all in benchmark suite 2.

Experiments were conducted on an AMD Ryzen 5 1600 Six-Core CPU with 16 GB of RAM running on Windows 10.
The results for second benchmark suite is shown Table \ref{tab:sra}. The timeout for all benchmarks is 300 seconds.

In the first benchmarks suite we are looking to verify a weaker form of the permutation property of sorting as shown in Section \ref{sec:mot_ex}.
Furthermore, we verify properties of two self-stabilizing algorithms for mutual exclusion on parameterized systems.
The first one is Lamport's bakery algorithm~\cite{lamport74}, for which we proved that the algorithm ensures mutual exclusion. 
The system is modelled in the style of regular model checking~\cite{RMC}, with
system states represented as words, here over an infinite alphabet:
the character representing a thread stores the thread control state, a Boolean flag, and an integer as the number drawn by the thread. The system transitions are modelled as parametric transducers, and invariants as parametric automata.
The second algorithm is known as Dijkstra's Self-Stabilizing Protocol \cite{dijkstra74}, in which system states are encoded as sequences of integers, and in which we verify that the set of states in which exactly one processor is privileged forms an invariant. The mentioned benchmarks require universal quantification, but similar to the motivating example from Section \ref{sec:mot_ex} one can eliminate quantifiers by Skolemization and instantiation which was done by hand.

The second benchmark suite consists of three different types of benchmarks, summarized in Table~\ref{tab:sra}.
The benchmark PR-C$n$ describes a regular expression for matching products which have the same code number of length $n$, and PR-CL$n$ matches not only the code number but also the lot number.
The last type of benchmark is IP-$n$, which matches $n$ positions of 2 IP addresses.
The benchmarks are taken from the regular-expression crowd-sourcing website RegExLib \cite{regex} and are also used in experiments for symbolic register automata\cite{SRA} which we also compare our results against.
To apply our decision procedure to the benchmarks,  we encode each of the benchmarks as a parametric automaton, using parameters for the (bounded-size) back-references.
The task in the experiments is to check emptiness, language equivalence, and language inclusion for the same combinations of the benchmarks as considered in \cite{SRA}.

\paragraph{Results of the Experiments}

\begin{table}[t]
	\centering
	\resizebox{\columnwidth}{!}{
		\begin{tabular}{|c c||c c|c c|c c|} 
			\hline
			$\mathcal{L}_1$&$\mathcal{L}_2$ & $ SRA_{\emptyset}(\mathcal{L}_1)$ &{\seqSolver}$_{\emptyset}(\mathcal{L}_1)$ & $ SRA_{\equiv}(\mathcal{L}_1)$ & {\seqSolver}$_{\equiv}(\mathcal{L}_1)$ & $SRA_{\subseteq}(\mathcal{L}_2,\mathcal{L}_1)$& {\seqSolver}$_{\subseteq}(\mathcal{L}_2,\mathcal{L}_1)$\\ [0.5ex] 
			\hline
			Pr-C2 & Pr-CL2 & 0.03s & 0.65s & 0.43s & 0.10s & 4.7s & 0.10s\\
			Pr-C3 & Pr-CL3 & 0.58s & 0.70s & 10.73s & 0.12s & 36.90s & 0.10s\\
			Pr-C4 & Pr-CL4 & 18.40s & 0.77s & 98.38s & 0.14s & - & 0.10s\\
			Pr-C6 & Pr-CL6 & - & 1.00s & - & 0.12s & - & 0.10s\\
			Pr-CL2 & Pr-C2 & 0.33s & 0.30s & 1.03s & 0.13s & 0.52s & 0.76s\\
			Pr-CL3 & Pr-C3 & 14.04s & 0.38s & 20.44s & 0.13s & 10.52s & 0.76s\\
			Pr-CL4 & Pr-C4 & - & 0.41s & 0.43s & 0.12s & - & 0.82s\\
			Pr-CL6 & Pr-C6 & - & 0.62s & 0.43s & 0.12s & - & 1.27s\\
			IP-2 & IP-3 & 0.11s & 1.53s & 0.63s & 0.14s & 2.43s & 0.15s\\
			IP-3 & IP-4 & 1.83s & 1.45s & 4.66s & 0.14s & 28.60s & 0.17s\\
			IP-4 & IP-6 & 30.33s & 1.75s & 80.03s & 0.14s & - & 0.17s\\
			IP-6 & IP-9 & - & 1.60s & 0.43s & 0.13s & - & 0.17s\\
			
			\hline
		\end{tabular}
	}
	\caption{Benchmark suite 2. $SRA$ is used for the algorithm for symbolic register automata and $SEQ$ for our tool. The symbol $\emptyset$ indicates the column where emptiness was checked, $\equiv$ indicates self equivalence and $\subseteq$ inclusion of languages.}\label{tab:sra}
\end{table}

All properties can be encoded by parametric automata with very few states and parameters. As a result the properties for each program can be verified in $<$ 2.6s, in detail the property for Dijkstra's algorithm was proven in 0.6s, QuickSort in 1.1s and Lamport's bakery algorithm in 2.5s.

%Similar to the experiments done for symbolic register automata we only consider pairs of benchmarks which are meaningful. 

The results for the second benchmark suite are shown in Table~\ref{tab:sra}.
The algorithm for symbolic register automata times out on 11 of the 36 benchmarks and our tool solves most benchmarks in $<$ 1s.
One thing to observe that the symbolic register automata scales poorly when more registers are needed to capture the back-references while the performance of our approach does not change noticeably when more parameters are introduced.

%\begin{table}
%\begin{center}
%	\begin{tabular}{|c | c|}
%		\hline
%		Dijkstra & 0.6s\\
%		Bakery & 2.5s\\
%		Quick Sort & 1.1s\\
%		\hline
%	\end{tabular} 
%\end{center}
%\caption{Benchmark suite 2. $SRA$ is used for the algorithm for symbolic register automata and $SEQ$ for our tool. The symbol $\emptyset$ indicates the column where emptiness was checked, $\equiv$ indicates self equivalence and $\subseteq$ inclusion of languages.}\label{tab:protocol}
%\end{table}

\section{Conclusion and Future Work}
\label{sec:concl}

In this paper, we have performed a systematic investigation of decidability and
complexity of constraints on sequences. Our starting point is the subcase of
string constraints (i.e. over a finite set of sequence elements), which include
equational constraints with concatenation, regular constraints, length
constraints, and transducers. We have identified parametric automata 
(extending symbolic automata and variable automata) as suitable notion of
``regular constraints'' over sequences, and parametric transducers (extending 
symbolic transducers) as suitable notion of transducers over sequences. We
showed that decidability results in the case of strings carry over to sequences,
although the complexity is in general higher than in the case of strings (sometimes
exponentially higher). 
For certain element theory (e.g. Linear Real Arithmetic),
it is possible to retain the same complexity as in the string case. 
We also delineate the boundary of the suitable notion of
``regular constraints'' by showing that the equational constraints with
symbolic register automata \cite{SRA} yields undecidable satisfiability.
Finally, our new sequence solver {\seqSolver} shows promising experimental 
results.

There are several future research avenues. Firstly, the complexity of sequence
constraints over other specific element theories (e.g. Linear Integer
Arithmetic) should be precisely determined. Secondly, is it possible to recover
decidability with other fragments of register automata (e.g., single-use automata
\cite{BS20})? On the implementation side, there are some algorithmic 
improvements, e.g., better nonemptiness checks for parametric automata in the
case of a single automaton, as well as product of multiple automata.

\paragraph*{Acknowledgment.} We thank anonymous reviewers for their thorough and
helpful feedback. We are grateful to Nikolaj Bj{\o}rner, Rupak Majumdar and 
Margus Veanes for the inspiring discussion.

\bibliographystyle{splncs04}
\bibliography{refs}

\appendix
\clearpage

\begin{center}
    \LARGE \textbf{APPENDIX}
\end{center}

\section{Proofs of Section~\ref{sec:model}}

\begin{proof}[of Lemma~\ref{lm:reduce}]
We show how to first remove an element constraint $\varphi$ that use the
constants $a_1, \ldots, a_k$. We simply create a new regular constraint $\Xseq \in
L(\Aut)$ that uses a fresh variable $\Xseq$ with the parametric automaton $\Aut$ 
with parameters $\params = \{a_1,\ldots,a_k\}$. The automaton
$\Aut$ recognizes precisely the set of sequences of the form $a_1,\ldots,a_k$
such that $\varphi$ is true. That is, $\Aut$ will have $k+1$ states
$q_0, \ldots, q_{k+1}$, where $q_0$ (resp.\ $q_{k+1}$) is the initial (resp.\ final)
state. The transition $q_i \to_{\psi} q_{i+1}$ uses the transition $\curr =
p_{i+1} \wedge \varphi$. It is easy to see that this removes $\varphi$, while
preserving satisfiability.

To make each parameters in a regular constraint $\Xseq \in L(\Aut)$ ``local'', we 
can simply make them ``visible'' in the string $\Xseq$. For simplicity, let us say
that $L(\Aut)$ has exactly one parameter $p$. We can devise a new automaton
using a fresh parameter $p'$ and accepts precisely the set of all sequences of
the form $p'.w$, where $w \in L_\mu(\Aut)$ for some $\mu$ instantiating $p$. The
resulting constraint becomes $\Yseq.\Xseq \in L(\Aut') \wedge \Yseq \in
L(\Aut'')$, where $\Aut''$ is the automaton that enforces that $\Yseq$ has
length 1. Constraints relating parameters across different parametric automata
can then be encoding using element constraints, which as we saw above could be
removed completely by means of a single parametric automaton using only 
``local'' parameters.
\end{proof}

\section{Proofs of Section \ref{sec:we}}
\begin{proof}[of Lemma~\ref{lem:finite_domain_subset_satisfiability}]
	If the constraints are satisfiable over $D_\pi$ then they are clearly satisfiable over $D$ (a larger set),
	as the same assignment works.
	
	If the constraints are satisfiable for $D$,
	then we change the assignment.
	For shortness of notation, for a type $t$ by $a_t$ we denote the chosen letter of this type in $D_\pi$,
	i.e.\ $a_t \in D_\pi, \type(a_t) = t$.
	Given an assignment satisfying all constraints we replace each symbol
	$a$ with $a_{\type(a)}$.
	We claim that such assignment still satisfies all constraints.
	
	For the regular constraint, suppose that $\mu(\Xseq) \in L_\mu(\Aut)$,
	let $\mu'(\Xseq)$ be the assignment value after the replacement.
	Note that $L_\mu(\Aut) = L_{\mu'}(\Aut) = L_\pi(\Aut)$,
	as both $\mu, \mu', \pi$ coincide on the values assigned to parameters.
	Then $\mu'(\Xseq) \in L_{\mu'}(\Aut)$: the corresponding letters
	of $\mu(\Xseq)$ and $\mu'(\Xseq)$ are of the same type, so they satisfy the same guards
	in $\Aut$ and so an accepting path for $\mu(\Xseq)$ yields the same accepting
	path for $\mu'(\Xseq)$ and vice versa.
	
	For the equational constraints: first observe that $|\mu(\Xseq)| = |\mu'(\Xseq)|$,
	as the latter was obtained from the former by a letter-to-letter replacement.
	Consider an equation $L = R$.
	If the corresponding letters in $\mu(L), \mu (R)$
	were both obtained from the variables,
	then they were replaced at both sides with the same letters.
	If symbols at both sides come from the constants,
	then they are clearly not changed (and still equal).
	If one side comes from a constant in the equation, say $c$,
	and the other from the variable, say $\Xseq$,
	then in $\mu(\Xseq)$ at the corresponding position $\mu(\Xseq)$ has $c$.
	As $c$ is a constant in the equation, the ``$x = c$'' is an atom in $\Phi$
	and so it is in the type $\type(c)$ and so $c$ is the unique letter
	(in whole $D$) with this type
	and so $c$ in $\mu(\Xseq)$ is replaced with $c$ in $\mu'(\Xseq)$ and so the equation is still satisfied.
	\qed
\end{proof}

%\section{Proof of Lemma \ref{lem:parameters_type_verification}}

\begin{proof}[of Theorem~\ref{thm:sequence_constraints_expspace}]
	Observe that the formula~\eqref{eq:parameters_type_verification} is polynomial in the size of $\tau$,
	so at most exponential.
	As the formula is exponential-size, it can be verified in \expspace,
	assuming that the existential fragment of $T$ can be verified in \pspace.
	
	The algorithm for word equations with regular constraints
	runs in \pspace, assuming that we have a \pspace oracle for the input alphabet
	(which can be exponential-size);
	such oracle can be implemented: given a type $t$ represented as a bitvector
	we can verify, whether this bitvector is in $\tau$, in \pspace, by simple search.
	In particular, if the formula~\eqref{eq:parameters_type_verification} can be verified in \pspace
	and $|\tau|$ is polynomial size,
	then the sequence constraints can be decided in \pspace as well.
	\qed
\end{proof}

\begin{proof}[of Corollary~\ref{cor:LRA}]
	It is folklore knowledge that in this case there are polynomially many different types:
	once the parameters assignments $\pi$ is fixed, 
	each atom in $\Phi$ is equivalent to a comparison of $\curr$ (the read symbol)
	with a number (depending on $\pi$).
	There are linearly many such numbers $n_1, n_2, \ldots n_k$ and so there are $2k+1$ realizable types:
	either $\curr$ is equal to one of those numbers or it is strictly between some consecutive two.
	Hence the claim follows from Theorem~\ref{thm:sequence_constraints_expspace}. \qed
\end{proof}

\begin{proof}[of Theorem~\ref{thm:we-expspace}]
	The theory that we are employing is the theory of automatic structures,
	i.e.\ $D = \Sigma^*$ for some finite alphabet $\Sigma$ (to be specified later)
	and $\theory$ uses automatic relations, i.e.\ relations $R \subseteq D^k$
	which can, roughly speaking, by recognized by an NFA reading the $k$-tuples of letters.
	In our case we will use only $1$- and $2$-automatic relations.
	To avoid confusion, we will refer to the elements of $D$ as strings,
	so that hey are not confused with sequences, so elements of $D^*$.
	Similarly, we talk about DFAs $A, \ldots $ and parametric automata $\Aut, \ldots$.

	The \expspace-hard problem that we are considering is a variant of the tiling problem:
	given a word $w \in \Sigma_0^*$ of length $n$ and two sets of tiles $V, H \subseteq \Sigma_0^2$
	decide, whether there is $m \in \mathbb N$ and a tiling $f : \{0, \ldots 2^n-1\} \times \{0, \ldots, m\} \to \Sigma_0$,
	of which we think as a board $2^n \times m$,
	such that
	\begin{itemize}
		\item $f(1,1) f(2,1) \cdots f(2^n,1) = wB^{2^n-n}$, where $B \in \Sigma_0$ is designated symbol;
		\item for each $0 \leq i < 2^n$ and $0 \leq j \leq m$ we have $\left(f(i,j),f(i+1,j)\right) \in H$, $\left (f(i,j), f(i,j+1)\right ) \in V$,
		i.e.\ each two consecutive horizontal values are a tile from $H$ and each two horizontally consecutive are in $V$;
		\item $f(1,m) f(2,m) \cdots f(2^n,m) = \downarrow B^{2^n-1}$, where $\downarrow, B \in \Sigma_0$ are designated symbols.
	\end{itemize}
	This problem is easily seen to be a reformulation of existence of exponentially space bounded accepting computation of a Turing Machine on input $w$.

	In our case $\Sigma = \Sigma_0 \cup \{ 0, 1, \#, \$\}$, where $0, 1, \#, \$$ are as special symbols 
	which are assumed to be not in $\Sigma_0$.
	We will often employ counting from $0$ to $2^n-1$
	and by $(i)_n$ for $0 \leq i < 2^n$ we denote a language $w \Sigma_0 \#$ where $w \in \{0,1\}^n$ is a binary notation of $i$ that uses $n$-bits,
	with the least significant digit first (let's say left);
	if $w a \# \in (i)_n$ for $a \in \Sigma_0$ then we say that $w a \#$ encodes $a$ in $(i)_n$.
	The idea is that we construct parametric automata such that their intersection is nonempty if only and only if
	there is a tiling $f$ and
	the parameter $p$ satisfies
	\begin{align}
		\label{eq:parameter_symbol_goal}
		F_{i,j} &\text{ encodes } f(i,j) \text{ in } (i)_n &\text{for } 0 \leq i < 2^n \text{ and } 0 \leq j \leq m\\
		\label{eq:parameter_row_goal}
		w_j
		&\in 
		F_{0,j} F_{1,j} \cdots F_{2^n-1,j} &\text{for } 0 \leq j \leq m\\
		\label{eq:parameter_column_goal}
		p
		&=
		w_0 w_1 \ldots w_m \$
	\end{align}
	In other words, $w_j$ encodes the $j$-th row of the tiling
	and $p$ is the concatenation of such encodings over each row.

	We first show a (variant of) folklore fact, that using intersection of $n+1$ DFAs of size $\mathcal O(n)$ each
	we can ``count'' from $0$ to $2^n-1$.
	
	\begin{lemma}
		\label{lem:counting_using_automata_lemma}
		There are DFAs $A_0, A_1, A_2, \ldots, A_n$ of size $\mathcal O(n)$ each such that
		$$
		\bigcap_{i=1}^n L(A_i) = ((0)_n (1)_n \cdots (2^n-1)_n )^* \$ \enspace .
		$$
		Similar claim holds for 
		$$
		\bigcup_{i=0}^{2^n-1} ((i)_n)^* \$ 
		$$
	\end{lemma}
	\begin{proof}
		The DFA $A_0$ ensures that the string is of the form $(\{0,1\}^n\Sigma_0 \#)^*\$$, which can be done using $n+2$ states.
		In the following description we ignore the $\$$, which is used only as a terminating marker.

		The $k$-th automaton $A_k$ ensures that when reading $(i)_n$ it has $k$-least significant digits $1^{k-1}0$
		(or $1^{k-1}1$)
		then the next read $(j)_n$has $0^{k-1}1$ (or $0^{k-1}0$), i.e.\ it takes care of the carry onto the $k$-th position.
		To this end the automaton counts the positions modulo $n+2$ and it reads the $k$ least significant digits.
		If they are all $1$ then it stores the $k+1$st digit (in the states)
		and checks that the next read number has $k$ digits $0$ and then the updated next digit.
		The automaton accepts when all such checks were satisfied and it read full encodings.
		It is easy to see that the claim holds.

		In the second case the construction is similar, but this time the $A_k$ stores the $k$-th symbol (assuming that it is $0$ or $1$)
		of the string $(i)_n$ in the state and ensures that the $k$-th digits of the next $(j)_n$ is the same.
		\qed
	\end{proof}
	
	We construct linear number of parametric automata,
	one of them, call it $\Aut$, will use a single parameter $p$
	and the other will not use any parameters.
	We ensure that
	\begin{equation}
		\label{eq:parameter_reduction_value}
		p \in \left ((0)_n (1)_n \cdots (2^n-1)_n \right )^*\$ 
	\end{equation}
	
	By Lemma~\ref{lem:counting_using_automata_lemma} there are $k = \mathcal O(n)$ DFAs $A_1, \ldots, A_m$ that ensure this
	and we write a subformula $R_{A_1}(p) \land \cdots \land R_{A_1}(p)$, where $R_{A_i}(p)$ holds if and only if $A_i$ accepts $p$,
	into the guards labelling the transitions from the starting state in $\Aut$.
	
	%The $p$ should be viewed as the desired tiling: when $p = w_1 w_2\cdots w_m$ where $w_j \in ((0)_n \Sigma_0  \# \;\;  (1)_n \Sigma_0  \# \;\; \cdots (2^n-1)_n \Sigma_0  \# \;\;)$ for each $j$,
	%then $w_j$ encodes one row of the tiling and $(i)_n$ gives the number of the cell in the row
	%and so $b$ stored after $(i)_n$ in $w_j$ corresponds to $f(i,j)$.

	On the other hand, we want to enforce that the sequence $\Xseq$ in the intersection of parametric automata is of the form
	\begin{equation}
		\label{eq:encoding_of_consecutive_numbers_*}
		\Xseq \in (0)_n^*, (1)_n^*, \ldots, (2^n-1)_n^* , \$
	\end{equation}
	Enforcing $\Xseq \in (j_0)_n^*, (j_1)_n^*, \ldots, (j_{2^n-1})_n^* , \$$
	for some $j_0, \ldots, j_{2^n-1}$s can be done by the Lemma~\ref{lem:counting_using_automata_lemma}:
	the $\Aut$ has a guard $\phi(\curr)$ which checks that the read string $\curr$ is in the language $\bigcup_{i=0}^{2^n-1} (i)_n^*$.
	On the other hand enforcing that $j_0 = 0$ and $j_{i+1} = j_i+1$, which implies $j_i = i$,
	can be done by a construction similar to the one in Lemma~\ref{lem:counting_using_automata_lemma}, but on the level of parametric automata:
	we use $n$ parametric automata (without parameters) $\Aut_1, \Aut_2, \ldots, \Aut_n$.
	The $k$-th automaton $\Aut_k$ reads consecutive sequences $s_1, s_2, \ldots$.
	If the $k$-th least significant digits in $s_i$ are $1^{k-1}0$ (or $1^{k-1}1$)
	then in the next string $s_{i+1}$ $\Aut_k$ expects $k$ least significant digits $0^{k-1}1$ (or $0^{k-1}0$).
	The condition on least $k$-significant digits can be checked by a simple regular $1$-ary relation of $\mathcal O(k)$ size.
	$\Aut_k$ accepts when all the checks passed;
	additionally, the $\Aut_n$ accepts only directly after reading $1^n$ in some $s_i$.
	The argument for correctness is immediate.

	Now we want to ensue that $p$ defines appropriate tiling.
	Enforcing that $w$ is encoded in the first row is simple:
	when $w = w_1\cdots w_n$ then we construct a DFA verifying that $p$
	is in $\{0,1\}^* w_1 \# \{0,1\}^* w_2 \# \cdots \{0,1\}^* w_n \# \Sigma^*$
	and add the corresponding relation to guard from the initial states of parametric automaton.
	We should also verify the ending condition:
	a similar DFA checks that after seeing $0^n \downarrow$ verifies that it sees $\{0,1\}^n B \#$
	and accepts directly after seeing $1^n B \#$.
	Again, the corresponding relation is added to the initial transitions.

	The tiles from $H$ are easy to define: a DFA counts modulo $n+2$ and after seeing $\{0,1\}^n$ it stores the next symbol $a$ in its finite memory,
	reads $\#$, reads $n$ bits and reads $b$, checking whether $(a,b) \in H$ and replaces $a$ with $b$ in the memory.
	Also, there is no comparison if the read bitstring is $0^n$ (as this means that we begin reading the next row).
	The appropriate 1-automatic relation is added to the transition of the parametric automaton.

	The last and crucial are the $V$ tiles.
	This is enforced using a parametric automaton $\Aut$ with a single parameter $p$ running on a sequence $\Xseq$ as defined in~\eqref{eq:encoding_of_consecutive_numbers_*}.
	The $\Aut$ has two states: the final state (with no outgoing transitions) and the starting state.
	The starting state has two transitions, one for $\curr = \$$ that goes to the final state
	and the other transition that goes to itself, labelled with $\phi$.
	The idea is that when $\curr = (i)_n$ then $\phi$ verifies the $V$ tiles for $i$, i.e.\ in the $i$-th column of the encoded tiling.
	The formula $\phi$ uses a single automatic relation on $\curr, p$, defined by automaton $A$;
	recall that $\curr \in (i)_n^*$ for some $i$.
	The DFA $A$ reads strings $p$ and $\curr$, counting modulo $n+2$,
	when it reads $(j)_n$ in $p$ and each of the corresponding bit in $\curr$ is the same, i.e.\ $i = j$,
	then it stores the next symbol from $p$ in the finite memory, say it is $a$,
	and when it reads $(i)_n$ in $p$ for the next time, with $b$ being the next symbol,
	it checks, whether $(a,b) \in H$ and replace $a$ with $b$ in the finite memory.
	$\Aut$ accepts, when there was no error and it gets to $\$$,
	i.e.\ the above checks were done for each $0 \leq i < 2^n$.

	It remains to show the correctness of the construction.
	So suppose that a board $f: \{0,\ldots, 2^{n}-1\} \times \{0, \ldots m\} \to \Sigma_0$ exists,
	set $p$ as in~\eqref{eq:parameter_symbol_goal}--\eqref{eq:parameter_column_goal}
	and set
	%\label{eq:parametr_row}
	%w_j
	%	&\in 
	%F_{0,j} \;\;  F_{1,j} \;\; \cdots F_{2^n-1,j}\\
	%\label{eq:parametr_column}
	%p
	%	&=
	%w_0 w_1 \ldots w_m \$\\
	\[
	\Xseq = (0)_n^{2^n m} , (1)_n^{2^n m} , \ldots, (2^n-1)_n^{2^n m} \$
	\]

	It is easy to see that they are accepted by all the defined parametric automata.
	
	So suppose that the intersection of the defined automata is non-empty.
	As already observed, this means that $p$ is of the form described in~\eqref{eq:parameter_reduction_value},
	hence it is of the form as in~\eqref{eq:parameter_symbol_goal}--\eqref{eq:parameter_column_goal} for some $m$.
	We define the valuation of the function as written in \eqref{eq:parameter_symbol_goal}--\eqref{eq:parameter_column_goal},
	i.e.\ $f(i,j)$ is the value after string from $\{0,1\}^n$ in $F_{i,j}$.
	
	One of the automatic relations explicitly takes care that $f(0,0)f(1,0)\cdots f(2^n-1,0) = w B^{2^n-n}$.
	Another that $f(0,m)f(1,m)\cdots f(2^n-1,m) = \downarrow B^{2^n-1}$.
	Also the horizontal tiles are taken care of by another automatic relation.
	What is left to show are the vertical tiles.
	Observe that it was already shown that
	$$
	\Xseq = (0)_n^{m_0} , (1)_n^{m_1} , \ldots, (2^n-1)_n^{m_{2^n-1}} \$
	$$
	for some $m_0, m_1, \ldots, m_{2^n-1}$.
	Since the parametric automaton accepts $\Xseq$, the automatic relation holds for $((i)_n^{m_i}, p)$
	for each $0 \leq i < 2^n$.
	This means that the DFA on pairs reads to the end of $p$, so in particular $m_i \geq 2^n m$.
	Moreover, by the way this relation is defined we can see that it verifies that $(f(i,j),f(i,j+1)) \in V$,
	so also all vertical constraints are verified. \qed
\end{proof}

As observed in Lemma~\ref{lm:reduce}, the element constraints can be reduced to regular and equational constraints.
However, when no regular constraints are present, then verifying element and equational constraints essentially boils down to
a verification of the element constraints (so in the theory $\theory$) and solving the equational constraints,
which can be interpreted as word equations.
In the end, the whole procedure is much simpler and has lower complexity.

%\begin{theorem}
%\label{thm:no_reg_constraints}
%For a theory $\theory$ decidable in \pspace,
%the element and equational constraints (so no regular constraints) can be decided in \pspace.
%\end{theorem}
\begin{proof}[of Theorem~\ref{thm:no_reg_constraints}]
We first collect all element constraints, let $\phi(\mathcal C)$ be their conjunction.
We iterate over possible partitions of $\mathcal C$ into $\mathcal C_1, \ldots, \mathcal C_k$
into sets that are equal, let $\psi(\mathcal C_1, \ldots, \mathcal C_k)$ be a formula which specifies those equalities
(e.g., for $\mathcal C_1 = \{c_1, \ldots, c_\ell\}$ it has $c_1 = c_2 \land c_2 = c_3 \land \cdots \land c_{\ell-1} = c_\ell$ subformula).
We verify the satisfiability of the element constraints,
i.e.\[
\exists \mathcal C \phi(\mathcal C) \land \psi(\mathcal C_1, \ldots, \mathcal C_k)
\]
It is is satisfiable, we solve the word equations over the alphabet containing letters from the equational constraints and one representative
for each class $\mathcal C_i$;
note that we explicitly identify all letters in one equivalence class.

It is easy to see that the procedure runs in \pspace and it will terminate, if there is a solution of the original system of constraints.

The constructed formula can be seen as an easier equivalent of~\eqref{eq:parameters_type_verification}.
\qed
\end{proof}
\section{Proofs in Section \ref{sec:sl}}

\begin{proof}[of Proposition~\ref{prop:pre-image}]
	The pre-image computation for parametric transducers can be done as follows.
    Let $\Tra = (\params, \controls, \Delta, q_0, \finals)$ be a parametric transducer
    and $\Aut = (\params, \controls', \Delta', q_0', \finals')$ a parametric automaton.
    Then we can define the pre-image of $\Aut$ as $\Aut' =
    (\params, \controls \times \controls', \transrel'',q_0, \finals \times
    \finals', )$. For every $(q_1, (\phi, \Wseq), q_2) \in \Delta$ we compute 
    all sequences of transitions with length $n := |\Wseq|$ in $\mathcal{A}$ and 
    for each sequence we add one new transition to $\transrel''$ which 
    simulates the entire sequence. More precisely, for a sequence of 
    $\transrel'$-transitions 
    \[
        (q_1, \phi_1), \dots (q_n, \phi_n)
    \]
    and one $\transrel-$transition $(p_1, (\psi, \textbf{w}), p_2) \in R$, we 
    add the transition 
    \[
        ((p_1, q_1), \psi  \wedge \phi_1[w_1/\curr] \wedge \dots \wedge
        \phi_n[w_n/\curr], (p_2, q_n)),
    \]
    where $\textbf{w} = (w_1,\dots, w_n)$. Note that each $w_i$ is of the form
    $t(\curr)$ for some term $t$. Here, $\phi_1[w_i/\curr]$ means
    $\phi_i$ but after applying substitution of every occurrence $\curr$ in 
    $\phi_i$ by $w_i$.  An automaton can have exponentially many paths of
    length $n$, hence the resulting automaton has exponentially many 
    transitions.
    \qed
\end{proof}

\begin{proof}[of Theorem \ref{thm:sl-expspace}]
    We now give the complexity analysis of the procedure. From the Proof of
    Proposition \ref{prop:pre-image}, each transition in any parametric
    automaton in any intermediate regular constraint $\Aut( \Xseq)$ will have
    as a guard a conjunction of formulas of the form
    $\psi(t(\curr))$, where $\psi$ is a guard occurring in a parametric
    automaton/transducer in the \emph{original} constraint and $t$ is of the
    form $t_1t_2\cdots t_m$ for some $m$ smaller than the number of assignments
    in $S$, and each $t_i(\curr)$ is a term in the parametric transducer in
    the \emph{original} constraint. Note that $t_1\cdots t_m$ simply means a
    composition of substitutions, e.g., $t_1 = \curr + 7$ and $t_2 = \curr + 10$,
    then $t_1t_2 = \curr + 17$. Therefore, by simple counting (and thinking of a
    conjunction as a set, after removing redundant conjuncts), there are at 
    most double exponentially many possible transitions in $\Aut$. 

    Furthermore, computing the pre-image of a transducer increases the number of
    states by at most a polynomial, while computing the pre-image of a
    concatenation does not increase the number of states. So, $\Aut$ will have 
    at most exponentially many states. We finally count the number of possible
    $\Aut$ that are generated in the intermediate steps. Computing the pre-image
    of a transducer does not generate new regular constraints, while
    concatenation yields a \emph{disjunction} of at most quadratically many
    regular constraints. Notice that we have to ``nondeterministically'' 
    choose one of these disjuncts each time. Nondeterminism is not needed
    since the number of such choices is at most exponential. Therefore, we can
    simulate this by a 2-EXPTIME deterministic machine.

    After removing all the assignments $\Yseq :=
    f(\Xseq_1,\ldots,\Xseq_r,\params)$, we end up with a conjunction of
    polynomially many regular constraints $\Aut(\Xseq)$, each having at most
    exponentially many states and doubly exponentially many transitions. By
    simply taking products, we obtain a parametric transducer $\mathcal{B}$ with at most
    exponentially many states and double exponentially many transitions. We can
    guess a simple path in $\mathcal{B}$, which has at most exponentially
    many states, resulting in a formula of exponential size. Since
    $T$ is in PSPACE, we obtain the double exponential time upper bound.
    \qed
\end{proof}

\section{Proofs of Section \ref{sec:ext}}

\begin{proof}[of Theorem~\ref{thm:length_constraints_simple}]
	Clearly, a word equation with regular and length constraint can be solved using a $T$-sequence solver
	(for regular, equational and length constraints).
	
	In the other direction:
	As outlined above, we can use Lemma~\ref{lem:parameters_type_verification} also when the length constraints are present.
	We iterate over possible sets of types $\tau \subseteq 2^\Phi$,
	for each one we verify the formula~\eqref{eq:parameters_type_verification}
	and if it is satisfiable, we check the validity of system of word equations with regular and length constraints.
	If it is satisfiable, then the original system was also satisfiable:
	the argument for equational and regular constraints is as in Lemma~\ref{lem:parameters_type_verification}
	and the same Lemma also claims that the lengths of solutions for the original and reduced systems are the same
	(the solution is in fact the same, just interpreted differently).
	If the original system was satisfiable for an assignment of parameters $\pi$,
	then it is satisfiable for $\tau = \{\type_\pi(a) \: : \: a \in D\}$ and so when we consider $\tau$,
	the corresponding system of word equations (with regular and length constraints) is satisfiable
	for the same $\nu$, see Lemma~\ref{lem:parameters_type_verification},
	and also formula~\eqref{eq:parameters_type_verification} is satisfied for assignment $\pi$.
	Lastly, the lengths are preserved, so also the length constraints are satisfied.
	\qed
\end{proof}

\begin{proof}[of Theorem~\ref{thm:ext-length}]
%	As the set of comparisons allowed in the formulas used in parametric automata is closed under the negation,
%	we assume that the formulas labelling the transitions are positive, i.e.\ negation is not used;
%	this is achieved by pushing the negation to the atoms and changing the atoms appropriately.

	We define $\type_\pi(a) = \{ \phi \in \Phi \: : \:  \phi(\pi(\params), a, \pi(\mathcal L)) \text{ holds}\}$.
	We iterate over possible sets of types $\tau \subseteq 2^\Phi$.
	For a given $\tau$ we introduce fresh integer variables $\{\ell_t, a_t\}_{t \in \tau}$
	those represent letters, but in the following we would not like to mix the variables used in
	an analogue of \eqref{eq:parameters_type_verification},
	which are treated as length-variables of artificial variables, so $\{\ell_t\}_{t \in \tau}$,
	and the letters used in word equations, so $\{a_t\}_{t \in \tau}$, which are treated as symbols with no specific meaning.
	We treat $\{\ell_t\}_{t \in \tau}, \params$ as length variables (of some fresh word variables).
	
	Consider the set of atoms
	\begin{equation}
		\label{eq:set_of_formulas_reduction}
		\Phi' = \bigcup_{t \in \tau} \bigcup_{\phi \text{ atom in } t} \{\phi(\params, \ell_t, \mathcal L)\} \enspace .
	\end{equation}
%	and their atoms; by assumption each $\phi$ is quantifier-free and without negation.
	Observe that if each atom in $\Phi'$ is satisfied
	then also each guard in the type from $\tau$ is satisfied.
	Note that $\Phi'$ can be seen as a system of ``length constraints''
	that use original length variables $\mathcal L$ and fresh ``length variables'' $\params \cup \{\ell_t \}_{t \in \tau}$.
	Note that there are no corresponding word variables, we can introduce them by dummy equations $X = X$.

	We consider a system of word equations (the input system of equational constraint) over alphabet $D_\pi$
	plus regular constraints, which are obtained exactly as in Lemma~\ref{lem:finite_domain_subset_satisfiability},
	plus original length constraints
	and fresh ``length constraints'' $\Phi'$.
	If this system is satisfiable, then also the original system is.

	Note that we have separated the letter $a_t$ used in word equations with its ``length'' $\ell_t$.
	However, the constructions guarantees that $a_t$ can use exactly the same transitions as $\ell_t$ in the parametric automaton.

	It remains to show that the reduction is valid.

	Suppose that some of the obtained system of word equations (with parameter constraints and length constraints) is satisfiable.
	We claim that the same values of sequence variables, in which $a_t$ is replaced with the value of $\ell_t$,
	satisfies all original constraints.
	Clearly it satisfies the equality constraints and length constraints, as those are the same.
	For regular constraints observe, that for $a_t$ we allowed the transitions by $\phi$
	when atoms of $\phi$ are in set~\eqref{eq:set_of_formulas_reduction}
	(recall that we assume that guard is a conjunction of atoms)
	and the set of new length constraints includes atoms that make such guard satisfied.
	Hence the corresponding transition can be made for $\ell_t$ in the parametric automaton (for assignment $\mu$)
	and hence the corresponding path in $\Aut$ exists.

	Suppose that the original constraints are satisfied,
	the proof is similar as in the case of Theorem~\ref{thm:length_constraints_simple}.
	Then the original constraints are satisfied by some $\mu$.
	Let $\tau = \{\type_\mu(a) \: : \: a \in D\}$ be the set of realisable types.
	We can change $\mu$ on $\mathcal V, \params$ so that it uses only one letter $a_t$ of a given type $t \in \tau$. i.e.\ from $D_\mu$:
	the argument is as in Lemma~\ref{lem:finite_domain_subset_satisfiability}.
	Consider, what the reduction does for $\tau$.
	Then all atoms in $\Phi'$ are satisfied by $\mu(\params), \{a_t\}_{t \in \tau}, \mu(\mathcal L)$.
%	clearly they make all formulas in~\eqref{eq:set_of_formulas_reduction} satisfied,
%	as those formulas are positive.
	The system of word equations is clearly satisfied (by $\mu$),
	so are the length constraints, as those are the same equations and same length constraints.
	The construction of the regular constraints guarantees that the transitions for $a_t$ in $\Aut$
	and in the constructed system are the same,
	hence the regular constraints are satisfied.
	Lastly, the new ``length constraints'' are satisfied by the values $\mu(\params), \{a_t\}_{t \in \tau}, \lambda(\mathcal L)$,
	by the choice of those atoms.
	\qed
\end{proof}

\section{Additional material: Section~\ref{subsec:register_automata} Undecidability of register automata constraints}

First, observe that we can assume that we have a finite number of chosen constants, called $a, b, c$, among letters in $D$:
to this end we introduce sequence variables $a, b, c$ --- 
one for each constant, and regular constraint $L = \{xy \: : \: x, y \in D, x\neq y, |x| = |y| = 1\}$, which can be clearly realized by a register automaton,
and writing constraints $\Aut(ab), \Aut(ac), \Aut(bc)$; formally this is realized by new sequence variables.

We give a construction, which results in $X, X', X''$ satisfying the following conditions:
\begin{enumerate}
\item $X'' = aX$ and $X''$ begins with $a$ and has all letters different; \label{enum:X}
\item $X' = (aX)^{|X|}$. \label{enum:X'}
\end{enumerate}
In the following we will use simple variants of this construction to encode
(positive) integer arithmetic with addition and multiplication,
with $X$ as above representing an integer variable with value $|X|$.

For a variable $X$ introduce a variable $X''$ and write a constraint
$X'' = aX$ and an automaton constraint that $X''$ has no other occurrence of $a$.
Write constraint $X''X' = X'X''$. It is well known that this implies that there is $aw \in D^+$ such that $X'' = (aw)^k$, $X' = (aw)^\ell$ for some $k, \ell$.
Since $X''$ has only a single $a$, $k = 1$ and so $X' = (aX)^\ell$.

We construct another register automaton and put a constraint on $X'a$
(formally this requires a new variable);
there is a special case when $X' = a$, which is trivially handled separately;
hence in the following we assume that $|w| > 0$.
The register automaton reads the first letter of $X'a = (aX)^k a$, i.e. $a$,
and stores it in the register $r_1$.
Then it enters a loop: it reads a value, stores it in $r_2$ and scans the input until it finds another occurrence of value from $r_2$.
If in the meantime it finds no letter $a$ (stored in $r_1$), then it rejects.
After finding another copy of $r_2$ it goes to the next element, stores it in $r_2$ and continues the loop
(the $r_1$ is not altered).

The loop ends when after finding the copy of $r_2$ the next letter is $a$,
in which case we accept (and reject in all other cases).

\begin{lemma}
\label{lem:register_automata_all_letters_different}
The construction above yields $X, X', X''$ satisfying conditions \ref{enum:X}, \ref{enum:X'}.
\end{lemma}
\begin{proof}
	By the construction we have that $X = w$ with $a$ not being a letter in $w$.
	It was argued that $X' = (aw)^\ell$ for some $\ell$.
	
	We call the automaton action between storing for the $i$-th time the letter
	in $r_2$ and finding its next occurrence an $i$-th pass.
	
	Suppose that $aw$ consists of different letters.
	By simple induction in $i$-th pass we store the $i$-th letter of $w$ in the register and read till $i$-th letter in the $i+1$-st copy of $w$ 
	(and store the next letter).
	Hence, in $w$-th pass we reach the last letter of $|w|$-th copy of $aw$,
	which is followed by $a$, so we accept after reading $(aw)^{|w|}a$,
	as claimed.
	We reject in each other case, so in particular for other powers $(aw)^na$, where $n \neq |w|$.

	If $X = w$ and $w$ contains a letter twice, say at positions $i < j$ then at the $i$-th pass we will reject, as there is no $a$ between those positions.
	\qed
\end{proof}

\begin{lemma}
\label{lem:register_automata_addition}
Using register automata constraints and sequence constraints we can enforce that substitution for $X, Y, Z$ has all letters different and $|X| + |Y| = |Z|$
\end{lemma}
Simply write $XY=Z$ and use the constraint from Lemma~\ref{lem:register_automata_all_letters_different} to $Z$, which in particular implies that all letters in $X, Y$ are different (and other than $a$).

We can use a variant of construction from Lemma~\ref{lem:register_automata_all_letters_different}
to simulate multiplication:
roughly, we use three variables $X, Y, Z$ and consider an equation
$(aXbYcZ)W = W(aXbYcZ)$ and make three types of checks in parallel.
For $Z$ we use the same construction as before, for $Y$ we ``reset'' after each $|Y|$ full passes
and for $X$ we make one pass for each full pass of $Y$.
This will ensure that $|Z| = |X| \cdot |Y|$.
\begin{lemma}
\label{lem:register_automata_multiplication}
Using register automata constraints and sequence constraints we can enforce that substitution for $X, Y, Z$ has all letters different and $|X| \cdot |Y| = |Z|$
\end{lemma}
\begin{proof}
	The proof is similar as in Lemma~\ref{lem:register_automata_all_letters_different}:
	Choose constant $a, b, c$, take new variables $X', Y', Z', W$, using regular constraints enforce that $a,b,c$ do not appear in $X, Y, Z$.
	Impose a sequence constraint
	$$
	(aXbYcZ) W = W (aXbYcZ)
	$$
	as in Lemma~\ref{lem:register_automata_all_letters_different} this implies that $W = (aXbYcZ)^k$ for some $k$.
	
	We construct a register automaton, which has three ``components'',
	working in parallel;
	each is similar to the automaton from Lemma~\ref{lem:register_automata_all_letters_different}.
	Formally the constraint is on sequence $abc W a$,
	so that constants $a, b, c$ are known to it and there is an $a$-terminator at the end.

	The first component works as Lemma~\ref{lem:register_automata_all_letters_different} for $Z$:
	i.e.\ it scans consecutive copies of $Z$, but it ignores letters between $a, c$ (including $c$, excluding $a$). Hence, it enforces that $k = |Z|$.
	
	The second component works similarly for $Y$ (ignoring letters between $c$ and $b$),
	but once it finds $c$ directly after a letter it stores in the register,
	it waits for $a$ to appear, on which it goes to an accepting state
	and then restarts, i.e.\ it waits for $b$ and then starts the computation again.
	Hence it is in an accepting state for each $(aXbYcZ)^{\ell|Y|}a$.
	We refer to the computation between accepting states as full pass,
	i.e.\ one full pass corresponds to $|Y|$ copies of $Y$ read.
	
	The third component performs the computation for $X$, but does one pass for $X$
	and then waits for the component responsible for $Y$ to go to an accepting state,
	in which case it makes another pass, and so on.
	Hence, it makes one pass per full pass for $Y$.
	Once an accepting state is reached, it remains there until another full pas of $Y$ is completed. Afterwards, it goes to a rejecting state.
	In this way, we ensure that $k = |X| \cdot |Y|$,
	hence $|Z| = |X| \cdot |Y|$.
	\qed
\end{proof}

Lemmata~\ref{lem:register_automata_all_letters_different}, \ref{lem:register_automata_addition}, \ref{lem:register_automata_multiplication}
straightforwardly allow encoding the Hilbert's tenth problem,
and so show the undecidability of sequence constraints and register constraints over infinite domains.

%\section{Proof of Lemma \ref{lem:register_automata_all_letters_different}}

%\section{Proof of Lemma \ref{lem:register_automata_multiplication}}

\end{document}